\documentclass[final,3p,times]{elsarticle}
\usepackage{amsmath,amssymb,amsthm,amsfonts}
\usepackage{graphicx,color,booktabs,array}
\usepackage{float}
\newfloat{algorithm}{htbp}{loa}\floatname{algorithm}{Algorithm}
\usepackage{tikz}
\usetikzlibrary{arrows.meta,decorations.markings,calc}
\usepackage{lmodern}
\usepackage{microtype}
\usepackage{rotating} % landscape page for Table 1

\usepackage[hidelinks]{hyperref}

\newtheorem{thm}{Theorem}[section]
\newtheorem{lem}[thm]{Lemma}
\theoremstyle{remark}
\newtheorem{rem}{Remark}[section]
\numberwithin{equation}{section}

\newcommand{\vu}{\mathbf u}
\newcommand{\vx}{\mathbf x}
\newcommand{\vv}{\mathbf v}
\newcommand{\vn}{\mathbf n}
\newcommand{\ey}{\mathbf e_y}
\newcommand{\tu}{\tilde{\mathbf u}}
\newcommand{\bu}{\bar{\mathbf u}}
\newcommand{\dt}{\delta t}
\newcommand{\G}{\Gamma}
\newcommand{\curl}{\nabla\times}

\newcommand{\RacFS}{657.510}
\newcommand{\RacRG}{1707.76}

\usepackage{lineno} % with \linenumbers after \begin{document}
\journal{Journal of Computational Physics}

\begin{document}
%\linenumbers

\begin{frontmatter}

\title{A second-order linear decoupled scheme with a relaxed zero-energy variable for Rayleigh--B\'enard convection}

\author{M. McCurdy}
 \ead{matthew.mccurdy@uvi.edu} 
\affiliation{organization={Department of Mathematical Sciences, University of the Virgin Islands},
 city={St. Thomas},
 state={USVI},
 country={United States }}

\begin{abstract}
We construct a second-order, linear, and fully decoupled scheme for Rayleigh--B\'enard convection between Navier-slip walls.
The nonlinear terms are treated explicitly and multiplied by a scalar with exact value $1$ and no energy contributed to the continuous problem. As a result, each BDF2 step, with a rotational pressure correction, solves only constant-coefficient linear problems.
For the time discretization, we prove that, for any time-step, the energy decays when the walls have equal temperatures, under a condition on the thermal weight, and the energy is bounded on finite time intervals when heated from below.
Relaxing the scalar towards $1$ preserves both estimates.
Energy stability does not control the accuracy of the scalar; the velocity and temperature fields are computed with its unrelaxed value, whose deviation from $1$ scales like $\frac{E}{C}\big(\frac{\dt}{\tau}\big)^3$ for kinetic energy $E$, time-scale $\tau$ of the flow, and weight $C$ of the scalar equation.
In our tests, choosing $C$ comparable to $E$ keeps this deviation small and doubles the largest time-step for which steady rolls are reproduced to a relative difference of $10^{-8}$.
With Taylor--Hood elements, we observe second-order convergence in time, recover the critical Rayleigh numbers, and match published Nusselt numbers to within $0.5\%$.
A three-dimensional implementation passes corresponding validation tests.
\end{abstract}

%% Research highlights
\begin{highlights}
\item Linear, decoupled BDF2 scheme for Rayleigh--B\'enard convection with Navier-slip walls
\item A relaxed zero-energy scalar absorbs the energy error of explicit nonlinear terms
\item Energy decay without heating, finite-time bound with heating, for any time-step
\item Energy stability does not control the accuracy of the zero-energy scalar
\item In our tests, a kinetic-energy weight doubles the largest accurate time-step
\end{highlights}


\end{frontmatter}

\section{Introduction}\label{sec:intro}
Energy-stable linear schemes are traditionally designed for dissipative systems, in which the energy can only decrease \cite{SXY2018,SXY2019}.
Rayleigh--B\'enard convection is not dissipative: heat enters continuously through the lower wall, and above a critical Rayleigh number the buoyancy feeds the flow \cite{Chandrasekhar1961,Rayleigh1916}.
In this work, we investigate what a zero-energy auxiliary variable, which makes every time-step linear and decoupled, still controls when it is applied to this continuously-driven problem, and what it does not.
We consider walls with a Navier-slip condition, which models finite wall friction between the free-slip and no-slip limits \cite{MCYZ2017}.
For flows like this, fully implicit schemes solve nonlinear systems at every time-step, linearly implicit schemes solve linear systems whose matrices change at every step, and an explicit treatment of the convection restricts the time-step.

To overcome this bottleneck, Scalar Auxiliary Variable (SAV) and Generalized SAV (GSAV) methods \cite{SXY2018,SXY2019} have emerged as powerful frameworks.
They enable the construction of linear, unconditionally energy-stable discretizations that reduce nonlinear PDE systems to a sequence of decoupled, constant-coefficient linear solves \cite{JiangLi2024, JiangYang2024jmaa, LiShenLiu2022, LinYangDong2019, Yang2021nme}.
However, existing SAV extensions to Boussinesq systems \cite{AlhomsiWuZheng2026, ChenHanWangZhang2020, JiangYang2023anm, WagnerWohlmuthZawallich2026} consider settings without heat input through the boundary: insulated walls, homogeneous temperature conditions, or perturbations of a stably stratified state.
In these cases, total thermal energy decays monotonically or cancels in weighted norm estimates.
Classical Rayleigh--B\'enard convection represents a fundamentally different physical regime though. Continuous thermal energy input from the heated lower wall acts as an active, sustained source, breaking global energy decay and limiting stability bounds to finite time intervals $[0,T]$.
Consequently, SAV schemes for continuous thermal-driving are needed, along with an understanding of what their energy estimates do and do not control.
Here, we address this challenge by developing a linear, second-order, fully decoupled scheme combining a rotational pressure correction \cite{GMS2006}, a zero-energy auxiliary scalar \cite{Yang2021nme}, and a scalar relaxation \cite{JiangZhangZhao2022} specifically tailored to Rayleigh--B\'enard convection under Navier-slip boundaries.

The scalar auxiliary variable (SAV) approach \cite{SXY2018,SXY2019} yields linear, unconditionally energy-stable schemes for dissipative systems by treating nonlinear terms explicitly and controlling them through a scalar variable tied to the energy of the system.
For incompressible flow, Lin, Yang, and Dong \cite{LinYangDong2019} applied this idea to the convection term, and Yang \cite{Yang2021nme} introduced a variant for coupled systems in which the scalar carries no energy, also called the zero-energy-contribution (ZEC) approach.
The approach of \cite{Yang2021nme} multiplies all nonlinear couplings (whose contributions to the energy balance cancel) and the scalar function's exact value is $1$.
The resulting schemes are fully decoupled and solve only constant-coefficient problems.
Scalar variables with a known exact value $e^{-t/T}$ are used by Li, Shen, and Liu \cite{LiShenLiu2022}, who combine an SAV with first-order and rotational second-order pressure correction for the Navier--Stokes equations, and by N.~Jiang and coauthors in ensemble methods, where all members share one constant coefficient matrix (with artificial compressibility for the Navier--Stokes equations \cite{JiangYang2024jmaa}, and with rotational pressure correction for the Stokes--Darcy system \cite{JiangLi2024}). 
Li, Shen, and Wang \cite{LiShenWang2026} take a different route for the Navier--Stokes--Darcy model.
Each step predicts the velocity and temperature fields with explicit convection and interface coupling, then rescales them by a factor computed from the energy of the model. 
With this, this energy is dissipated, and the solutions stay bounded.
They prove a first-order error estimate and extend the approach to two-phase flow.

For the Boussinesq equations, Jiang and Yang \cite{JiangYang2023anm} proposed stabilized SAV ensemble schemes of Crank--Nicolson leap-frog and BDF2-type.
They introduce separate auxiliary scalars for the momentum and the heat equations, treat all nonlinear and coupling terms explicitly, and prove long-time stability for homogeneous Dirichlet conditions.
Two recent works use the GSAV approach, in which an energy-based scalar is advanced by an exponential integrator: Wagner, Wohlmuth, and Zawallich \cite{WagnerWohlmuthZawallich2026} for no-slip walls with homogeneous temperature conditions, and Alhomsi, Wu, and Zheng \cite{AlhomsiWuZheng2026} for perturbations of a stably stratified rest-state.
The rescaling of \cite{LiShenWang2026}, focused on the coupled Navier--Stokes--Darcy system, is likewise built on an energy that the model dissipates.
In all three works, no heat enters through the walls. 
With homogeneous conditions, the temperature decays and the resulting buoyancy work is bounded.
As a result, for stable stratification, the buoyancy and stratification terms cancel in a weighted energy.
For a two-phase Navier--Stokes--Darcy--Boussinesq system with insulated boundaries, Chen, Han, Wang, and Zhang \cite{ChenHanWangZhang2020} weight the thermal energy differently. They scale by a constant large enough that the buoyancy work (bounded with Poincar\'e and Korn inequalities) is absorbed by the viscous and thermal dissipation.
Rayleigh--B\'enard convection, though, differs from those; heat enters through the bottom wall, the buoyancy work is an energy source that no weight cancels or absorbs, the no-flow steady-state is unstable above a critical Rayleigh number \cite{Chandrasekhar1961, Rayleigh1916}, and a stability estimate can at best bound the energy on finite time intervals.
Table~\ref{tab:works} below compares these works with the present one.

This paper constructs and tests a numerical scheme for the case of Rayleigh--B\'enard convection.
Its components have been studied and implemented independently of each other: BDF2 with the rotational pressure-correction method \cite{GMS2006}, the zero-energy variable $Q$ of Yang \cite{Yang2021nme} which multiplies the convection of momentum and the advection of heat, the relaxation of Jiang, Zhang, and Zhao \cite{JiangZhangZhao2022}, proposed for gradient flows, and walls with a Navier-slip condition, as in \cite{MCYZ2017}.
While the zero-energy variable, the pressure correction, and the relaxation are each already established, our contribution is to show how they behave together when heating from below supplies energy at every step.
And although the energy estimate holds for every time-step, the scalar that makes the scheme linear can become inaccurate, by an amount that grows with the ratio of the kinetic energy of the flow to the weight of the scalar equation.
To our knowledge, no previous SAV or GSAV scheme for the Boussinesq equations treats heating from below with Navier-slip walls (Table~\ref{tab:works}).
The novel contributions of this manuscript are the following:

\begin{enumerate}
\item \emph{A scheme for heating from below.} We combine these approaches into a second-order, linear, fully decoupled scheme for Rayleigh--B\'enard convection between Navier-slip walls (free slip for zero slip coefficient, no slip in the limit of an infinite one), with the temperature written as the conductive steady-state plus a perturbation which vanishes on the walls.
Each step reduces to four constant-coefficient Helmholtz problems (two for the velocity, two for the temperature), one pressure Poisson problem, and a scalar update.
The relaxation is adapted to the BDF2 energy of $Q$.
\item \emph{Energy analysis of the time discretization} (Theorem~\ref{thm:stab}) with exact spatial operators.
For any time-step, we prove decay of a discrete energy when the walls are at the same temperature, under a condition on the weight $A$ of the thermal energy, and a finite-time bound under heating, where no weight removes the energy source.
The analysis of the rotational form with Navier-slip walls rests on an identity for flat walls (Lemma~\ref{lem:walls}), and holds unchanged in three dimensions (Remark~\ref{rem:3d}).
We prove neither a fully discrete estimate nor an error estimate; Section~\ref{sec:noerr} states the obstructions.
\item \emph{Energy stability is not accuracy of the scalar.} The relaxation controls the drift of the stored scalar $Q^n$, but the velocity and temperature fields of a step are computed with the unrelaxed value $\tilde Q^{n+1}$, which the relaxation does not change.
We weight the equation of $Q$ by a constant $C$ with the dimension of an energy, which leaves the analysis unchanged, and argue that $\tilde Q^{n+1}-1=O\big(\frac{E}{C}\big(\frac{\dt}{\tau}\big)^3\big)$ for a flow with kinetic energy $E$ and time scale $\tau$.
The computations confirm the $\dt^3$ and $\frac{1}{C}$ dependence, and the growth with $E$ over two orders of magnitude of $E$, which gives the practical criterion $C\gtrsim E$.
The range of $\tilde Q^n$, not the stored $Q^n$, is the diagnostic that detects failure, and a limit of the explicit convection, $\dt\max|\nabla\vu|\approx1.2$--$1.6$, remains.
\end{enumerate}

We implement the scheme (Section~\ref{sec:impl}) with Taylor--Hood finite elements and remove the hydrostatic pressure of the conduction state analytically.
This is essential near onset, since these elements are not pressure robust.
The numerical experiments (Section~\ref{sec:num}) begin with a manufactured solution showing second-order convergence in time and the expected spatial orders of the discretization.
The onset of convection checks the linear behavior against the exact free-slip growth rates and the critical Rayleigh numbers for free-slip and rigid walls.
Steady rolls check the nonlinear heat transport against weakly nonlinear theory \cite{MalkusVeronis1958} and published Nusselt numbers for rigid \cite{CleverBusse1974} and free-slip \cite{MooreWeiss1973} walls.
Runs at $Ra=5000$ with different time-steps and weights $C$, and at Rayleigh numbers from $2000$ to $2\cdot10^4$, test the accuracy of the scalar and its scaling with $\frac{E}{C}$.
We also check the decay of the discrete energy without heating and the sensitivity to the relaxation parameter.
A run at $Ra=10^6$ serves as a stress test, and Section~\ref{sec:3d} verifies that the scheme and its energy behavior carry over to three dimensions.
The FreeFEM \cite{Hecht2012} scripts, the command line of every run reported in this work, and the Python scripts that produce the tables and figures from the run logs are provided as supplementary material.

The numerical experiments expose three limitations on accuracy that the energy estimate does not capture.
Without relaxation, $Q$ can drift away from $1$. 
As a result, the computed steady-states are those of a different problem, with the convection and advection scaled by the drifted value of $Q$ (in Section~\ref{sec:dt}, the Nusselt number of steady rolls is off by $7\%$ at $\dt=5\cdot10^{-4}$ with $C=1$).
Additionally, for larger time-steps, $Q$ drifting from $1$ switches off the advection that saturates the instability, so that the solution grows without bound although the energy estimate holds.
With $C=1$, the energy error of the explicit terms is large compared with the energy of $Q$, and $\tilde Q$ collapses within a step.
With $C$ comparable to the kinetic energy, $\tilde Q$ stays within about $1\%$ of $1$.
The observed convergence during the formation of the rolls is then second order, with smaller errors.
Steady rolls are also reproduced up to a time-step twice as large, where the limit of the explicit convection takes over.

 \begin{sidewaystable}
\centering
\smallskip
\setlength{\tabcolsep}{3pt}\scriptsize
\renewcommand{\arraystretch}{1.12}
\begin{tabular}{>{\raggedright\arraybackslash}p{0.73in} >{\raggedright\arraybackslash}p{1.18in}
>{\raggedright\arraybackslash}p{0.68in} >{\raggedright\arraybackslash}p{0.78in}
>{\raggedright\arraybackslash}p{1.18in} >{\raggedright\arraybackslash}p{1.24in}
>{\raggedright\arraybackslash}p{1.16in}}
\toprule
Work & System; walls & Time; order & Pressure & Scalar variable(s) & Proved & Space; tests\\
\midrule
\multicolumn{7}{l}{\emph{Navier--Stokes and related flows}}\\
\cite{LinYangDong2019} & NS; Dirichlet & BDF1/2; 2 & decoupled (two Poisson problems) & one energy-based, $\sqrt{E}$; nonlinear scalar equation (Newton) & energy law & spectral elements; manufactured, cylinder\\
\cite{LiShenLiu2022} & NS; no slip & BE; 1\newline BDF2; 2 & standard PC\newline rotational PC & one, exact value $e^{-t/T}$; linear & unconditional stability; error estimate for BE (2D) & MAC grid; manufactured, cavity\\
\cite{JiangYang2024jmaa} & NS ensembles; no slip & CNLF; 2 & artificial compressibility & one per member, $e^{-t/T}$; stabilized & long-time stability under a condition on the parameter spread & Taylor--Hood FE; manufactured, cylinder, 3D cavity\\
\cite{JiangLi2024} & Stokes--Darcy ensembles & BE; 1 & rotational PC (tunable) & one per member, $e^{-t/T}$, for the interface terms & long-time stability under three parameter conditions & FE; manufactured, stochastic, realistic\\
\cite{LiShenWang2026} & NS--Darcy (and CH--NS--Darcy); Beavers--Joseph--Saffman interface & BE; 1\newline CN/AB2; 2 & coupled Stokes solve, constant coefficients & none stored: the predicted fields are rescaled by $\sqrt{\xi}$, $\xi^{n+1}=E^n/(\tilde E^{n+1}+\dt\,I)$ & dissipation of the original energy, boundedness; error estimate for BE (2D) & manufactured, filtration\\
\cite{Yang2021nme} & NS + mass-conserving AC, two-phase & 2 & decoupled & one zero-energy scalar for all couplings, exact value $1$ & unconditional stability & \\
\cite{JiangZhangZhao2022} & gradient flows (no flow) & BDF2, CN; 2 & -- & energy-based, relaxed towards its exact value & unconditional stability & Fourier\\
\midrule
\multicolumn{7}{l}{\emph{Boussinesq systems}}\\
\cite{ChenHanWangZhang2020} & CH--NS--Darcy--Boussinesq; no slip, no heat input & 1 & coupled (NS), stabilized (Darcy) & none (convex splitting) & unconditional stability (weighted energy) & Taylor--Hood FE; B\'enard cells\\
\cite{JiangYang2023anm} & Boussinesq ensembles; no slip, no heat input & CNLF, BDF2; 2 & coupled & two per member ($\vu$, $\theta$), $e^{-t/T}$; stabilized; quadratic scalar equation & unconditional long-time stability & FE; manufactured, Marsigli flow\\
\cite{WagnerWohlmuthZawallich2026} & Boussinesq; no slip, no heat input & BDF-type; 2 & consistent splitting & one energy-based (GSAV), exponential update & bounds on $[0,T]$; error estimate $O(\dt^2)$ (2D) & FE; manufactured, Marsigli flow (2D, 3D)\\
\cite{AlhomsiWuZheng2026} & Boussinesq, perturbation of stable stratification; no slip & BDF2-type; 2 & consistent splitting & one energy-based (GSAV) & bounds on $[0,T]$; error estimate $O(\dt^2)$ (2D) & manufactured, internal waves\\
\midrule
This paper & Boussinesq, heated from below; Navier slip (free slip to no slip) & BDF2; 2 & rotational PC & one zero-energy scalar, weighted by $C$, exact value $1$, relaxed; linear & decay without heating; bound on $[0,T]$ with heating (any $\dt$); no error estimate & Taylor--Hood FE; manufactured, onset, published $Nu$, energy decay\\
\bottomrule
\end{tabular}
\caption{Linear schemes with scalar variables for incompressible flow and Boussinesq-type systems, and the
works whose techniques are used here. ``Proved'' lists what is proved for the time discretization (for the
walls in column 2); ``no heat input'' means that the temperature satisfies homogeneous conditions or the
walls are insulated, so that no heat enters through the walls. BE: backward Euler; CNLF: Crank--Nicolson
leap-frog; AC: Allen--Cahn; CH: Cahn--Hilliard; PC: pressure correction; FE: finite elements. Order: order of accuracy in time.}
\label{tab:works}
\end{sidewaystable}

\newpage

\section{Model}\label{sec:model}
For notation, $(f,g)=\int_{\Omega} fg\,d\vx$ and $\| f \|^2 = (f,f)$.
The $x$-, $y$-, and $z$-components of $\vu$ are $u_x,$ $u_y,$ and $u_z,$ respectively.

Let $\Omega=[0,L)\times(0,H)$, periodic in $x$, with walls $\G=\G_b\cup\G_t$ at $y=0$ and $y=H$ (Figure~\ref{fig:domain}(a)).
With the Navier--Stokes--Boussinesq approximation, we have
\begin{align}\label{pde}
\left\{
\begin{array}{l}
\vu_t+(\vu\cdot\nabla)\vu=\nu\Delta\vu-\nabla p+G\,(T-T_{\rm ref})\,\ey\,\\
 \nabla\cdot\vu=0\,,\\
T_t+\vu\cdot\nabla T=\kappa\Delta T\,,
\end{array}
\right.
\end{align}
where $\vu$ is the fluid velocity, $p$ the pressure divided by the constant reference density $\rho_0$, $T$ the temperature, $\nu$ the kinematic viscosity, and $\kappa$ the thermal diffusivity of the fluid.
In the Boussinesq approximation, the density is taken constant except in the gravity term, where it depends linearly on the temperature, $\rho=\rho_0\big(1-\beta(T-T_{\rm ref})\big)$, with $\beta$ as the coefficient of thermal expansion and $T_{\rm ref}$ as a reference temperature.
Here, $G=g\beta$ with $g$ as the acceleration due to gravity, which points along $-\ey$, where $\ey$ is the upward-pointing unit vector.
The bottom and top walls are held at the temperatures $T_b$ and $T_t$, with $H$ as the depth of the layer and $L$ as the length of the periodic cell.
We take $T_{\rm ref}$ between the wall temperatures: $T_{\rm ref}=\frac{1}{2}$ for $T_b=1$, $T_t=0$ in Section~\ref{sec:num}.
(We note that using a different value for $T_{\rm ref}$ only adds a linear function of $y$ to the pressure.)
The Rayleigh number $Ra=G(T_b-T_t)H^3/(\nu\kappa)$ measures the strength of the buoyancy against viscous and thermal diffusion, and the Prandtl number $Pr=\nu/\kappa$ compares the diffusion of momentum with that of heat.
The numerical experiments of Section~\ref{sec:num} are specified by $Ra$ and $Pr$.

On the walls, we enforce no-penetration and a Navier-slip condition for the flow along with Dirichlet conditions for temperature:
\begin{equation}\label{bc}
\vu\cdot\vn=0\quad \textrm{and} \quad l\,u_x+\nu\,\partial_\vn u_x=0 \quad \textrm{on $\G$}\,,\qquad T=T_b\ \text{on }\G_b,\quad T=T_t\ \text{on }\G_t,
\end{equation}
where $\vn$ is the outward-pointing unit normal vector and the slip coefficient $l\ge0$: $l=0$ gives free-slip walls and $l\to\infty$ no-slip walls.
The same Navier-slip condition is used in other linear, energy-stable schemes, like the Allen--Cahn--Navier--Stokes model of two-phase flow with moving contact lines \cite{MCYZ2017}, where slip at the wall is what allows the contact line to move.

\begin{figure}[htbp]
\centering
\resizebox{\textwidth}{!}{\input{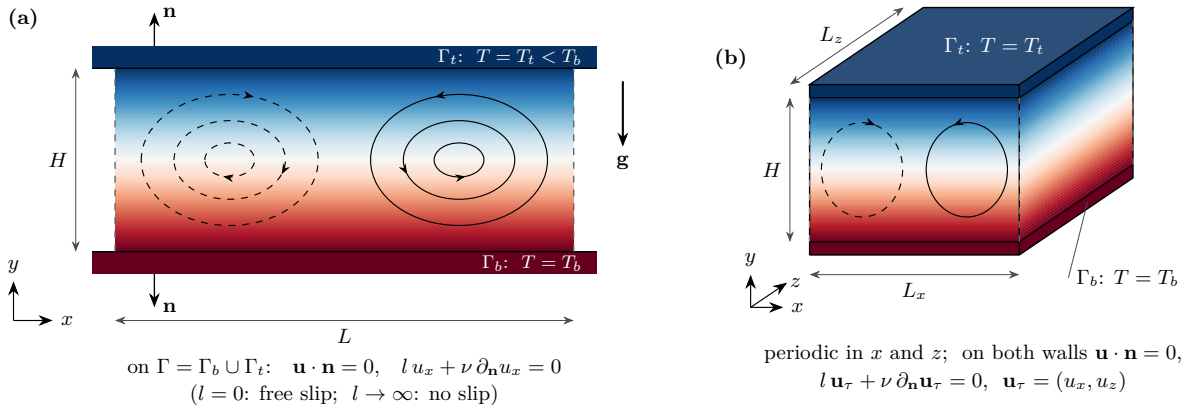}}
\caption{The domain. (a) Periodic cell $\Omega=[0,L)\times(0,H)$, heated from below ($T_b>T_t$), with gravity along $-\ey$.
The walls carry the Navier-slip condition \eqref{bc}, and the dashed sides $x=0$ and $x=L$ are noted.
Shading is the conductive linear profile $T_c(y)$, and the sketch shows two counter-rotating rolls, with warm fluid rising at $x=0$ and cold fluid sinking at $x=L/2$ (dashed: clockwise, as in Figure~\ref{fig:flows}). (b) The three-dimensional channel of Remark~\ref{rem:3d} and Section~\ref{sec:3d}, periodic in $x$ and $z$.}
\label{fig:domain}
\end{figure}

\paragraph{Conductive steady-state and perturbation} The conductive steady-state solution is the fluid at rest $\vu=0$ with a linear temperature profile $T_c(y)=T_b+T_c'\,y$, where $T_c'=(T_t-T_b)/H$, balanced by a hydrostatic pressure.
We write the temperature as this steady-state plus a perturbation, $T=T_c+\theta$.
The perturbation $\theta$ vanishes on $\G$, since $T_c$ already takes the wall values.
Substituting the $T=T_c+\theta$ expansion into the advection-diffusion equation of \eqref{pde} gives the equation governing $\theta$:
\begin{equation}\label{theta}
\theta_t+\vu\cdot\nabla\theta+T_c'\,u_y=\kappa\Delta\theta .
\end{equation}

\paragraph{Energy balance} For a weight $A>0$, summing the inner products of the momentum equation of \eqref{pde} with $\vu$ and \eqref{theta} taken with $A\theta$ gives:
\begin{equation}\label{law}
\frac{d}{dt}\left(\frac{1}{2}\|\vu\|^2+\frac{A}{2}\|\theta\|^2\right)
=-\nu\|\nabla\vu\|^2-l\|u_x\|_\G^2-A\kappa\|\nabla\theta\|^2+\big(G(T_c+\theta-T_{\rm ref}),u_y\big)-AT_c'(u_y,\theta),
\end{equation}
since $((\vu\cdot\nabla)\vu,\vu)=(\vu\cdot\nabla\theta,\theta)=0$ for $\nabla\cdot\vu=0$, $\vu\cdot\vn=0$.
The last two terms of \eqref{law}---the work of the buoyancy and the exchange of energy with the conduction profile, respectively---are sign-indefinite.

Since $T_c-T_{\rm ref}$ depends on $y$ only, $G(T_c-T_{\rm ref})\ey$ is a gradient and does no work on $\vu$, so the last two terms reduce to $(G-AT_c')(\theta,u_y)$.
When heated from above ($T_c'>0$), the choice $A=G/T_c'$ makes that term cancel.
When heated from below though ($T_c'<0$), we have that $G-AT_c'>0$ for every $A>0$, and the rising warmer fluid and sinking cooler fluid, $(\theta,u_y)>0$, feed the energy.
This is innate to convection, and $A$ cannot be chosen to remove that energy.
At Rayleigh numbers below the onset of convection, the energy method bounds this term by the dissipation and turns the balance into nonlinear stability of the no-flow conductive steady-state.
For Rayleigh--B\'enard convection, the resulting nonlinear threshold coincides with the linear one \cite{Straughan2004}.
For a fluid layer over a saturated porous medium, the linear and nonlinear stabilities lie close together \cite{McCurdyMooreWang2019}.
If the walls are at the reference temperature, $T_b=T_t=T_{\rm ref}$, only $(G\theta,u_y)$ remains.
Since $\theta$ and $u_y$ vanish on $\G$, the Poincar\'e inequality $\|v\|\le C_P\|\partial_yv\|$, with the sharp constant $C_P=H/\pi$, applies to both.
With Young's inequality $ab\le\epsilon a^2+\frac{b^2}{4\epsilon}$---taken with $\epsilon=\nu$, $a=\|\nabla\vu\|$, and $b=G\,C_P^2\|\nabla\theta\|$---we have
\begin{equation*}
|(G\theta,u_y)|\le G\|\theta\|\,\|u_y\|\le G\,C_P^2\|\nabla\theta\|\,\|\nabla\vu\|\le\nu\|\nabla\vu\|^2+\frac{G^2C_P^4}{4\nu}\|\nabla\theta\|^2 .
\end{equation*}
The choice $\epsilon=\nu$ lets the viscous dissipation $-\nu\|\nabla\vu\|^2$ in \eqref{law} absorb the first term, so that
\begin{equation*}
\frac{d}{dt}\left(\frac{1}{2}\|\vu\|^2+\frac{A}{2}\|\theta\|^2\right)\le-l\|u_x\|_\G^2-\Big(A\kappa-\frac{G^2C_P^4}{4\nu}\Big)\|\nabla\theta\|^2 .
\end{equation*}
Hence, the energy does not increase once $A\kappa\ge G^2C_P^4/(4\nu)$, as in \cite{ChenHanWangZhang2020}, where the boundaries are insulated.

\paragraph{Zero-energy variable} Following \cite{Yang2021nme}, we introduce the scalar variable $Q(t)$ with $Q(0)=1$ and fix a constant $C>0$ with the dimension of an energy.
The scalar $Q$ is multiplied by the nonlinear terms of the momentum equation and the advection-diffusion equation for $\theta$, and an equation for $Q$ is added:
\begin{align}\label{Qform}
\left\{
\begin{array}{l}
\vu_t+Q(\vu\cdot\nabla)\vu=\nu\Delta\vu-\nabla p+G(T_c+\theta-T_{\rm ref})\ey,\\
\nabla\cdot\vu=0,\\
\theta_t+Q\,\vu\cdot\nabla\theta+T_c'u_y=\kappa\Delta\theta,\\
C\,Q_t=((\vu\cdot\nabla)\vu,\vu)+A(\vu\cdot\nabla\theta,\theta).
\end{array}
\right.
\end{align}
The right-hand side of the $Q$ equation vanishes since $((\vu\cdot\nabla)\vu,\vu)=(\vu\cdot\nabla\theta,\theta)=0$ as noted with the energy balance above.
So, $Q\equiv1$ and \eqref{Qform} is equivalent
to \eqref{pde} for every $C>0$.
Multiplying the $Q$ equation by $Q$ and adding it to the energy
balance, we see the two terms multiplied by $Q$ cancel.
That is, $\frac{C}{2}Q^2$ is a constant that carries no physical energy.
In the continuous problem, $Q$ contributes nothing to the energy balance---in this sense, it is {\em zero-energy}.
Its BDF2 discretization, however, contributes $\frac{C}{4}(|Q^n|^2+|2Q^n-Q^{n-1}|^2)$ to the discrete energy of Section~\ref{sec:stab}, an auxiliary numerical energy that changes only through the energy error of the explicit terms.
The discrete scheme keeps this cancellation while still treating the nonlinear terms explicitly.
The original formulation \cite{Yang2021nme} has $C=1$, and our scheme recovers a result analogous to theirs (\cite{Yang2021nme} used a two-phase flow with no heat).
The role of $C$ in the discrete scheme is discussed after \eqref{relax}.

\section{Scheme}\label{sec:scheme}
Let $\dt>0$, $t_n=n\dt$, and write $D_tX^{n+1}=(3X^{n+1}-4X^n+X^{n-1})/(2\dt)$ for the BDF2 time integration.
We write the extrapolations of $\vu$ and $\theta$ as $\bu=2\vu^n-\vu^{n-1}$ and $\bar\theta=2\theta^n-\theta^{n-1}$, respectively.
Both the BDF2 scheme and the extrapolations are second-order approximations; their pairing in this scheme follows the convention of \cite{SXY2018}. \\

\noindent Given $(\vu^n,\vu^{n-1},p^n,\theta^n,\theta^{n-1},Q^n,Q^{n-1})$, with $\nabla\cdot\vu^k=0$ in $\Omega$ and $\vu^k\cdot\vn=0$ on $\G$ for $k=n,n-1$:

\medskip\noindent\textbf{Step 1.} We set up a system for $\tu^{n+1}$ and $\tilde Q^{n+1}$---intermediate steps on the way to $\vu^{n+1}$ and $Q^{n+1}$, respectively---and $\theta^{n+1}$.
We find $(\tu^{n+1},\theta^{n+1},\tilde Q^{n+1})$ such that
\begin{align}
&\frac{3\tu^{n+1}-4\vu^n+\vu^{n-1}}{2\dt}+\tilde Q^{n+1}(\bu\cdot\nabla)\bu
=\nu\Delta\tu^{n+1}-\nabla p^n+G(T_c+\bar\theta-T_{\rm ref})\ey,\label{su}\\[10pt]
&\qquad\textrm{with}\quad \tu^{n+1}\cdot\vn=0\quad \textrm{and}\quad l\tilde u_x^{n+1}+\nu\partial_\vn\tilde u_x^{n+1}=0\quad \textrm{on }\G,\notag\\[20pt]
&D_t\theta^{n+1}+\tilde Q^{n+1}\bu\cdot\nabla\bar\theta+T_c'\bar u_y=\kappa\Delta\theta^{n+1},
\qquad \textrm{with}\quad \theta^{n+1}=0\quad \textrm{on }\G,\label{sT}\\[20pt]
&C\,\frac{3\tilde Q^{n+1}-4Q^n+Q^{n-1}}{2\dt}=((\bu\cdot\nabla)\bu,\tu^{n+1})+A(\bu\cdot\nabla\bar\theta,\theta^{n+1}).\label{sQ}
\end{align}

\noindent\textbf{Step 2.} Next, we introduce a rotational pressure correction that splits each step into a Helmholtz problem for $\tu$ with the old pressure and a Poisson problem that projects onto divergence-free fields.
In the standard form, $p^{n+1} = p^n +\phi^{n+1}$ and the Neumann condition $\partial_\vn\phi=0$ freezes $\partial_\vn p$ on the wall at its initial value.
This artificial condition limits the accuracy of the pressure.
The {\em rotational form} $p^{n+1} = p^n +\phi^{n+1}-\nu \nabla \cdot \tu^{n+1}$ includes $-\nu \nabla \cdot \tu^{n+1}$ to remove the artificial condition.
This process, introduced in \cite{TimmermansMinevVandeVosse1996} and analyzed by Guermond and Shen \cite{GuermondShen2004}, avoids the artificial Neumann condition that the standard incremental form imposes on the pressure \cite{GMS2006}.
(Section~\ref{sec:mms} shows that with finite elements, the standard form loses accuracy on meshes that resolve the resulting numerical boundary layer at the walls.)
The same pressure correction appears in velocity-correction methods \cite{GuermondShen2003,ChenShen2019} and in decoupled schemes for coupled systems, such as the Cahn--Hilliard--Darcy--Stokes system \cite{WangZouLi2024} and with magnetohydrodynamics \cite{ZhangSuLiu2023}.
Pressure-correction schemes have also been analyzed with open boundary conditions \cite{GuermondMinevShen2005}, extended to variable viscosity with natural boundary conditions \cite{PlasmanDeteixYakoubi2020}, and used in ensemble schemes for the Boussinesq equations \cite{Jiang2019jsc}.

We find $\phi^{n+1}$ and $\vu^{n+1}$ such that
\begin{align}
&\frac{3(\vu^{n+1}-\tu^{n+1})}{2\dt}+\nabla\phi^{n+1}=0,\qquad \nabla\cdot\vu^{n+1}=0,\label{sp}\\[4pt]
&\qquad\textrm{with}\quad \vu^{n+1}\cdot\vn=0\quad\textrm{on }\G,\notag\\[12pt]
&p^{n+1}=p^n+\phi^{n+1}-\nu\nabla\cdot\tu^{n+1}.\label{spsp}
\end{align}
Taking the divergence of the first equation in \eqref{sp} and using $\nabla\cdot\vu^{n+1}=0$ gives the Poisson equation below, and its normal component on $\G$---where $\vu^{n+1}\cdot\vn=\tu^{n+1}\cdot\vn=0$---gives the Neumann condition:
\begin{equation}\label{sphi}
\begin{aligned}
&\Delta\phi^{n+1}=\frac{3}{2\dt}\nabla\cdot\tu^{n+1},\quad\textrm{with}\quad \partial_\vn\phi^{n+1}=0\quad\textrm{on }\G.
\end{aligned}
\end{equation}

\noindent\textbf{Step 3.} Finally, we relax the scalar, following \cite{JiangZhangZhao2022} (see also
\cite{zhang2022generalized}).
We store
\begin{equation}\label{relaxQ}
Q^{n+1}=\xi\,\tilde Q^{n+1}+(1-\xi),\qquad\text{or equivalently}\qquad Q^{n+1}-1=\xi\,\big(\tilde Q^{n+1}-1\big),
\end{equation}
where $\xi$ is the smallest number in $[0,1]$ such that
\begin{equation}\label{relax}
\begin{gathered}
\mathcal E_Q\big(\xi\,\tilde Q^{n+1}+1-\xi\big)-\mathcal E_Q(\tilde Q^{n+1})\le\eta\,\dt\,\mathcal D^{n+1},\\[4pt]
\textrm{with}\quad\mathcal E_Q(q)=\tfrac{C}{4}\big(q^2+(2q-Q^n)^2\big),\qquad
\mathcal D^{n+1}=\tfrac{\nu}{2}\|\curl\tu^{n+1}\|^2+l\|\tilde u_x^{n+1}\|_\G^2 .
\end{gathered}
\end{equation}
Here, $\mathcal E_Q(q)$ is the energy of the scalar at the new time level when its value is $q$, $\mathcal D^{n+1}$ is the viscous and wall dissipation of the step, and $\eta\in[0,1)$ is a fixed fraction.
We call $Q^{n+1}$ the \emph{stored scalar} and $\tilde Q^{n+1}$ the \emph{computational scalar}: only the stored value enters the next step, while the velocity and temperature of step $n+1$ are computed with $\tilde Q^{n+1}$.

Relaxation is needed because the computed value carries the energy error of the explicit terms.
Since $((\bu\cdot\nabla)\bu,\bu)=0$ and $(\bu\cdot\nabla\bar\theta,\bar\theta)=0$, the right-hand side of \eqref{sQ} equals
\begin{equation*}
R^{n+1}:=((\bu\cdot\nabla)\bu,\tu^{n+1}-\bu)+A(\bu\cdot\nabla\bar\theta,\theta^{n+1}-\bar\theta),
\end{equation*}
and $\tilde Q^{n+1}-1=\frac{2\dt}{3C}R^{n+1}$ when $Q^n=Q^{n-1}=1$.
If $\tilde Q^{n+1}$ were stored, these errors would accumulate through \eqref{sQ}, and $Q$ would drift to a value that rescales the convection and advection.
This could cause a computed steady-state to belong to a different problem ($Nu$ is off by $7\%$ at $\dt=5\cdot10^{-4}$ with $C=1$, Section~\ref{sec:dt}).

Moving $Q$ towards $1$ can raise the discrete energy of Section~\ref{sec:stab}, which contains $\mathcal E_Q(Q^{n+1})$.
Constraint \eqref{relax} pays for any such increase out of the fraction $\eta$ of the dissipation, so Theorem~\ref{thm:stab} holds for every admissible $\xi$ with $(1-\eta)\mathcal D^{n+1}$ to spare (which is why $\eta<1$).
The left-hand side of \eqref{relax} is a convex quadratic in $\xi$ that vanishes at $\xi=1$; so, the smallest admissible $\xi$, which stores the value closest to $1$, is either $0$ or the smaller root (Section~\ref{sec:impl}).
In practice though, the budget far exceeds the cost:
\begin{equation*}
\tfrac{3C}{2}\,\big|\tilde Q^{n+1}-1\big|=\dt\,|R^{n+1}|=O(\dt^3)
\qquad\text{against}\qquad
\eta\,\dt\,\mathcal D^{n+1}=O(\dt)
\end{equation*}
(Section~\ref{sec:notcontrolled}), so $Q^{n+1}$ is reset to exactly $1$ at every step unless the flow is nearly at rest.

For our simulations, we take $\eta=0.95$.
Every $\eta\in[0.5,1)$ gives the same results, while $\eta=0$ behaves like no relaxation (Section~\ref{sec:eta}).
The relaxation changes only the value stored for the next step: the velocity and temperature fields of step $n+1$ are computed with $\tilde Q^{n+1}$, whose deviation from $1$ is controlled by $\dt$ and $C$.

\paragraph{Remark} The first step uses backward Euler with the same structure ($\bu=\vu^0$,
$\bar\theta=\theta^0$).

\paragraph{The weight $C$} For $Q^n=Q^{n-1}=1$, \eqref{sQ} gives $\tilde Q^{n+1}-1=\frac{2\dt}{3C}R^{n+1}$, with $R^{n+1}$ the energy error of the explicit terms defined in Step~3.
The size of this deviation follows from a scaling argument, not from a bound.
Consider a flow with kinetic energy $E=\frac{1}{2}\|\vu\|^2$ (not the discrete energy $E^n$ of Section~\ref{sec:stab}), velocity scale $U=\sqrt{2E/|\Omega|}$ (from $E\sim \frac{1}{2}U^2$ averaged over a domain $\Omega$), and a single time-scale $\tau$, which is both the advective time $\ell/U$ of its length-scale $\ell$ and the time over which the flow changes.
The extrapolation error is $\tu^{n+1}-\bu=\dt^2\,\partial_{tt}\vu+O(\dt^3)$, of size $U\dt^2/\tau^2$, and $(\bu\cdot\nabla)\bu$ is of size $U^2/\ell=U/\tau$.
The first term of $R^{n+1}$ is therefore of size $|\Omega|\,\frac{U}{\tau}\,\frac{U\dt^2}{\tau^2}=\frac{2E}{\tau}\big(\frac{\dt}{\tau}\big)^2$, and the second has the same form with the thermal energy $\frac{A}{2}\|\theta\|^2$ in place of $E$.
With $\tilde Q^{n+1}-1=\frac{2\dt}{3C}R^{n+1}$,
$$|\tilde Q^{n+1}-1|\sim\frac{E}{C}\Big(\frac{\dt}{\tau}\Big)^3,$$
up to a constant that depends on the flow.
The factor $\dt^3$ comes from one $\dt$ in \eqref{sQ} and two from the extrapolation, and the factor $\frac{1}{C}$ comes from the weight of \eqref{sQ}.
The factor $E$ comes from $R^{n+1}$ being quadratic in the velocity: each of its two factors, the convection term and the extrapolation error, carries one factor of $U$, and $|\Omega|U^2=2E$.

The velocity and temperature fields of the intermediate step are computed with $\tilde Q^{n+1}$, and the relaxation---which acts on the stored value---does not reduce this deviation (Lemma~\ref{lem:contract}).
With $C=1$, this error is amplified by $E$, which in convection grows with the Rayleigh number (in the units of Section~\ref{sec:num}, $E\approx140$ for the rolls at $Ra=5000$, $E\approx400$ at $Ra=10^4$, and more in time-dependent convection).
We therefore take $C$ comparable to the kinetic energy of the flow, fixed during a run; equivalently, the variables can be scaled so that the energy is of order one.
The scaling of $\tilde Q^{n+1}-1$ with $\dt^3$ and $1/C$ is a prediction of this argument, and Section~\ref{sec:dt} confirms both: slope $3$ in $\dt$, and a reduction by the factor $100$ from $C=1$ to $100$ and again to $10^4$ (Figure~\ref{fig:C}(b)).
Over Rayleigh numbers from $2000$ to $2\cdot10^4$, it also confirms the growth with $E$ (Figure~\ref{fig:scaling}).
The criterion $C\gtrsim E$ is a practical rule suggested by this scaling and confirmed in the tests of Section~\ref{sec:dt}; it is not a condition of the stability analysis, which holds for every $C>0$.
Below the explicit-convection limit, taking $C$ larger than $E$ did no harm there ($C=10^4$ at $Ra=5000$, where $E\approx142$, gives the same results as $C=100$ up to $\dt=6\cdot10^{-3}$). Beyond that limit, at $\dt\ge10^{-2}$, the relaxation no longer reset the stored $Q$ to $1$ with $C=10^4$, while it did with $C=100$.
Neither the theorem nor the relaxation depends on $C$ (see the discussion after Lemma~\ref{lem:contract}).

An a priori value can therefore be taken from the free-fall velocity $U_{\rm ff}=\sqrt{G|T_b-T_t|H}$.
In all our heated runs that remained bounded, the kinetic energy stayed below $\frac{1}{2}|\Omega|U_{\rm ff}^2$.
This bound is $5\cdot10^3$ for the rolls at $Ra=5000$ ($E\approx142$) and $2\cdot10^6$ for the run at $Ra=10^6$ (largest $E\approx2.4\cdot10^5$).
Theorem~\ref{thm:stab} holds for every $C>0$ with the same constants, and $C=1$ recovers the formulation of \cite{Yang2021nme}. %We speculate the same consideration applies to any scalar variable whose equation absorbs the energy error of explicit terms, although we have not thoroughly examined it for other schemes.

\section{Energy stability}\label{sec:stab}
This section proves that, for every time-step, the energy of the time discretization of Section~\ref{sec:scheme} decays without heating, under a condition on the thermal weight $A$, and is bounded on finite time intervals with heating from below.
The analysis is for exact spatial operators, with the finite-element version discussed after the theorem.
Our proof has three components:
\begin{itemize}
\item two identities for flat walls with a Navier-slip condition, which turn the viscous term of the rotational form into the vorticity plus a wall term (Section~\ref{sec:walls}), as in the analysis of vorticity at Navier-slip walls \cite{ClopeauMikelicRobert1998,Kelliher2006};
\item a discrete energy $E^n$ that contains, besides the kinetic and thermal energies, a pressure term produced by the pressure correction, as in the stability analysis of the rotational form \cite{GuermondShen2004,GMS2006}, and the discrete energy of the scalar $Q$ (Section~\ref{sec:thm});
\item the cancellation of the explicit convection and advection terms against the equation of $Q$, which is what the zero-energy variable is for \cite{LinYangDong2019,Yang2021nme}.
\end{itemize}
Heating enters only through the work of the buoyancy.
Without heating, this work is absorbed by the dissipation (Theorem~\ref{thm:stab}(a)).
However, when heating from below, buoyancy becomes an energy source, and the estimate is a Gr\"onwall bound on finite time intervals (Theorem~\ref{thm:stab}(b)).
Section~\ref{sec:notcontrolled} states what the theorem does {\em not} give, Section~\ref{sec:noerr} details our hindrances to proving a more generalized error estimate, and Section~\ref{sec:stab3d} extends the result to three dimensions.

\subsection{Two identities for flat Navier walls}\label{sec:walls}
\begin{lem}\label{lem:walls}
Let $\vv\in H^1(\Omega)^2$ be periodic in $x$ with $v_y=0$ on $\G$, and let $\omega=\curl\vv=\partial_xv_y-\partial_yv_x$.
Then,
\begin{equation}\label{id1}
\|\nabla\vv\|^2=\|\omega\|^2+\|\nabla\cdot\vv\|^2 .
\end{equation}
If, moreover, $\vv\in H^2(\Omega)^2$ and $lv_x+\nu\partial_\vn v_x=0$ on $\G$, then
\begin{equation}\label{id2}
\nu(\curl\curl\vv,\vv)=\nu\|\omega\|^2+l\|v_x\|_\G^2 .
\end{equation}
\end{lem}
\begin{proof}
For \eqref{id1}, a direct computation gives
\begin{equation*}
|\nabla\vv|^2-\omega^2-(\nabla\cdot\vv)^2=2\big(\partial_xv_y\,\partial_yv_x-\partial_xv_x\,\partial_yv_y\big)
=2\partial_x\big(v_y\,\partial_yv_x\big)-2\partial_y\big(v_y\,\partial_xv_x\big).
\end{equation*}
The first term of the last expression integrates to zero by periodicity.
The second integrates to boundary terms that contain $v_y$, which vanishes on $\G$ (for $\vv\in H^1$, with $v_y=0$ on $\G$).

For \eqref{id2}, in two dimensions $\curl\curl\vv=(\partial_y\omega,-\partial_x\omega)$, and integration by parts gives
\begin{equation*}
(\curl\curl\vv,\vv)=\|\omega\|^2+\int_\G\omega\,v_x\,n_y\,ds .
\end{equation*}
On a wall, $v_y=0$ implies $\partial_xv_y=0$, so $\omega=-\partial_yv_x$.
The slip condition gives
$\partial_yv_x=\frac{l}{\nu} v_x$ at $y=0$, where $n_y=-1$, and $\partial_yv_x=-\frac{l}{\nu} v_x$ at $y=H$, where $n_y=1$.
In both cases, we have $\omega\,v_x\,n_y=\frac{l}{\nu} v_x^2$.
\end{proof}

Identity \eqref{id1} is classical for $\vv\in H^1_0(\Omega)^2$, \cite{GiraultRaviart1986}.
On flat walls, it still holds when only the normal component of $\vv$ vanishes.
Identity \eqref{id2} is the energy form of the boundary relation between vorticity and slip velocity used for Navier-slip walls in
\cite{ClopeauMikelicRobert1998,Kelliher2006}.
Both identities are used throughout this section.
The intermediate velocity $\tu^{n+1}$ satisfies the Navier-slip condition, but is not divergence free---\eqref{id2} gives its viscous dissipation.
The end-of-step velocity $\vu^{n+1}$ is divergence free, and \eqref{id1} gives $\|\nabla\vu^{n+1}\|=\|\curl\vu^{n+1}\|$.
Both identities use that the walls are flat (Section~\ref{sec:stab3d}).

\subsection{Discrete energy and main result}\label{sec:thm}
The rotational pressure correction puts a pressure term into the energy, as in the analyses of
\cite{GuermondShen2004,GMS2006}, and we measure it in the following norm.
Let
$S:=I-\frac{2\nu\dt}{3}\Delta$ with homogeneous Neumann conditions, and the $M$-norm be defined as:
\begin{equation}\label{Mnorm}
\|q\|_M^2:=(\nabla q,\nabla S^{-1}q).
\end{equation}
With the Neumann eigenpairs $-\Delta\psi_k=\lambda_k\psi_k$ and $c=\frac{2\nu\dt}{3}$, $(\nabla a,\nabla S^{-1}b)=\sum\limits_k\frac{\lambda_k}{1+c\lambda_k}a_kb_k$ for all $a,b\in H^1$.
This form is symmetric and nonnegative, and
\begin{equation}\label{Mprop}
\|q\|_M\le\|\nabla q\|,\qquad\|\nabla S^{-1}q\|\le\|q\|_M .
\end{equation}

Now, using the $M$-norm for pressure, the discrete energy is defined as
\begin{equation}\label{En}
\begin{aligned}
E^n={}&\underbrace{\tfrac{1}{4}\big(\|\vu^n\|^2+\|\bu\|^2\big)}_{\text{kinetic}}
+\underbrace{\tfrac{\dt^2}{3}\|p^n\|_M^2}_{\text{pressure}}+\underbrace{\tfrac{A}{4}\big(\|\theta^n\|^2+\|\bar\theta\|^2\big)}_{\text{thermal}}
+\underbrace{\tfrac{C}{4}\big(|Q^n|^2+|\bar Q|^2\big)}_{\text{scalar}} ,
\end{aligned}
\end{equation}
with the extrapolation for $\bar Q:= 2Q^n-Q^{n-1}$ (also second-order, like $\bu$ and $\bar\theta$).
The kinetic and thermal parts are the usual BDF2 energies \cite{SXY2018,GuermondShen2004}.
The scalar part equals $\frac{C}{2}$ when $Q^n=Q^{n-1}=1$; it is the discrete counterpart of the constant $\frac{C}{2}Q^2$ of Section~\ref{sec:model}, and it changes only through the energy error of the explicit terms.
We also use
\begin{itemize}
\item the dissipation budget of the relaxation, $\mathcal D^{n+1}=\frac{\nu}{2}\|\curl\tu^{n+1}\|^2+l\|\tilde u_x^{n+1}\|_\G^2$,
from \eqref{relax};
\item the Poincar\'e constant in $y$, $C_P=H/\pi$: $\|v\|\le C_P\|\partial_yv\|$ when $v=0$ on $\G$.
\end{itemize}

\begin{thm}\label{thm:stab}
Let $n\ge1$ and let $\dt>0$ be arbitrary.
\begin{enumerate}
\item[(a)] \emph{No heating.} Let $T_b=T_t=T_{\rm ref}$, and
\begin{equation*}
B':=\frac{G^2C_P^4}{2\nu}+\frac{\dt\,G^2C_P^2}{3},\qquad
X^n:=E^n+\dt B'\big(10\|\nabla\theta^n\|^2+2\|\nabla\theta^{n-1}\|^2\big).
\end{equation*}
If $A\kappa\ge10B'$, then
\begin{equation}\label{thma}
X^{n+1}-X^n\le-\dt\Big((1-\eta)\mathcal D^{n+1}+(A\kappa-10B')\|\nabla\theta^{n+1}\|^2\Big)\le0 .
\end{equation}
\item[(b)] \emph{Heating.} For any wall temperatures and any $A>0$,
\begin{equation}\label{thmb}
E^{n+1}\le(1+K\dt)E^n+c\,\dt,\qquad\text{hence}\qquad E^n\le e^{K t_n}\big(E^1+c\,t_n\big),
\end{equation}
with
\begin{equation*}
K:=\frac{4G^2}{A}\Big(\frac{C_P^2}{\nu}+\frac{2\dt}{3}\Big)+\frac{2A|T_c'|^2C_P^2}{\kappa},\qquad
c:=G^2\Big(\frac{C_P^2}{\nu}+\frac{2\dt}{3}\Big)\|T_c-T_{\rm ref}\|^2 .
\end{equation*}
\end{enumerate}
Both statements also hold without the relaxation; that is, with Step~3 of the scheme omitted and
$Q^{n+1}=\tilde Q^{n+1}$, if $\eta=0$ in \eqref{thma}.
\end{thm}

\begin{proof}
We write $\omega^{n+1}=\curl\tu^{n+1}$ and $\delta p=p^{n+1}-p^n$, and use the BDF2 identity
\begin{equation}\label{bdfid}
2(3a-4b+c)\,a=a^2+(2a-b)^2-b^2-(2b-c)^2+(a-2b+c)^2
\end{equation}
for scalars along with its analogue for inner products (see, e.g., \cite{GuermondShen2004,SXY2018}).
The last term, $(a-2b+c)^2$, is nonnegative; it is the numerical dissipation of BDF2.
We drop it each time \eqref{bdfid} is used (for $\vu$ in Step~2, for $\theta$ and $Q$ in Step~3), which is why those relations become inequalities.

\medskip\noindent\emph{Step 1: the momentum equation in rotational form.} By \eqref{sp}, $\tu^{n+1}-\vu^{n+1}=\frac{2\dt}{3}\nabla\phi^{n+1}$,
so $\nabla\cdot\tu^{n+1}=\frac{2\dt}{3}\Delta\phi^{n+1}$ (since $\nabla \cdot \vu^{n+1} = 0$).
By \eqref{spsp}, we then obtain
\begin{equation*}
\delta p=\phi^{n+1}-\nu\nabla\cdot\tu^{n+1}=S\phi^{n+1}.
\end{equation*}
Adding \eqref{su} and the first equation of \eqref{sp} cancels the terms in $\tu^{n+1}$ on the left and leaves $D_t\vu^{n+1}$:
\begin{align*}
D_t\vu^{n+1}+\tilde Q^{n+1}(\bu&\cdot\nabla)\bu=\nu\Delta\tu^{n+1}-\nabla p^n-\nabla\phi^{n+1}+\mathbf b,\\
&\quad\textrm{with}\quad \mathbf b:=G(T_c+\bar\theta-T_{\rm ref})\ey .
\end{align*}
We then rewrite the two terms on the right which involve $\tu^{n+1}$ and the pressure.
For the viscous term, the identity $-\Delta=\curl\curl-\nabla\nabla\cdot$ gives
\begin{equation*}
\nu\Delta\tu^{n+1}=-\nu\curl\curl\tu^{n+1}+\nu\nabla\nabla\cdot\tu^{n+1}.
\end{equation*}
For the pressure, \eqref{spsp} gives
\begin{equation*}
\nabla p^n+\nabla\phi^{n+1}=\nabla p^{n+1}+\nu\nabla\nabla\cdot\tu^{n+1}.
\end{equation*}
The two terms $\nu\nabla\nabla\cdot\tu^{n+1}$ cancel, and we are left with:
\begin{equation}\label{sum}
D_t\vu^{n+1}+\nu\curl\curl\tu^{n+1}+\nabla p^{n+1}=-\tilde Q^{n+1}(\bu\cdot\nabla)\bu+\mathbf b .
\end{equation}

\medskip\noindent\emph{Step 2: kinetic energy.} We take the inner product of \eqref{sum} with $\dt\,\tu^{n+1}$.
The left-hand side gives three terms: the time derivative, the viscous term, and the pressure.
We treat them in turn.
\begin{itemize}
\item \emph{time derivative.} $\tu^{n+1}-\vu^{n+1}$ is a gradient, and $D_t\vu^{n+1}$ is divergence free with zero normal component, so the two are orthogonal.
Using \eqref{bdfid} (with $a=\vu^{n+1}$, $b=\vu^n$, $c=\vu^{n-1}$) allows:          
\begin{equation*}
\begin{aligned}
\dt(D_t\vu^{n+1},\tu^{n+1})&=\dt(D_t\vu^{n+1},\vu^{n+1})\\
&=\tfrac{1}{4}\big(\|\vu^{n+1}\|^2+\|2\vu^{n+1}-\vu^n\|^2\big)\\
&\qquad-\tfrac{1}{4}\big(\|\vu^n\|^2+\|2\vu^n-\vu^{n-1}\|^2\big)\\
&\qquad+\tfrac{1}{4}\|\vu^{n+1}-2\vu^n+\vu^{n-1}\|^2\\[5pt]
&\ge\tfrac{1}{4}\big(\|\vu^{n+1}\|^2+\|2\vu^{n+1}-\vu^n\|^2\big)\\
&\qquad-\tfrac{1}{4}\big(\|\vu^n\|^2+\|2\vu^n-\vu^{n-1}\|^2\big).
\end{aligned}
\end{equation*}
\item \emph{viscous term.} By \eqref{id2},
$\dt\nu(\curl\curl\tu^{n+1},\tu^{n+1})=\dt\nu\|\omega^{n+1}\|^2+\dt l\|\tilde u_x^{n+1}\|_\G^2$.
\item \emph{pressure.} $(\nabla p^{n+1},\vu^{n+1})=0$ and $\nabla\phi^{n+1}=\nabla S^{-1}\delta p$.
The identity
$2(a,a-b)=|a|^2-|b|^2+|a-b|^2$, applied to the symmetric form of \eqref{Mnorm}, gives
\begin{equation*}
\begin{aligned}
\dt(\nabla p^{n+1},\tu^{n+1})&=\tfrac{2\dt^2}{3}\big(\nabla p^{n+1},\nabla S^{-1}\delta p\big)\\
&=\tfrac{\dt^2}{3}\big(\|p^{n+1}\|_M^2-\|p^n\|_M^2+\|\delta p\|_M^2\big).
\end{aligned}
\end{equation*}
\end{itemize}

\medskip\noindent\emph{Step 3: thermal energy and the scalar.} The inner product of \eqref{sT} with $A\dt\,\theta^{n+1}$ gives (by
\eqref{bdfid}, with the nonnegative term dropped):
\begin{multline*}
\tfrac{A}{4}\big(\|\theta^{n+1}\|^2+\|2\theta^{n+1}-\theta^n\|^2\big)-\tfrac{A}{4}\big(\|\theta^n\|^2+\|2\theta^n-\theta^{n-1}\|^2\big)
+A\dt\kappa\|\nabla\theta^{n+1}\|^2\\
\le-A\dt\,\tilde Q^{n+1}(\bu\cdot\nabla\bar\theta,\theta^{n+1})-A\dt\,T_c'(\bar u_y,\theta^{n+1}).
\end{multline*}
Equation \eqref{sQ} multiplied by $\dt\,\tilde Q^{n+1}$, again by \eqref{bdfid} and with $\mathcal E_Q$ from
\eqref{relax}, gives:
\begin{equation*}
\begin{aligned}
&\mathcal E_Q(\tilde Q^{n+1})-\tfrac{C}{4}\big(|Q^n|^2+|2Q^n-Q^{n-1}|^2\big)\\
&\qquad\le\dt\,\tilde Q^{n+1}\big[((\bu\cdot\nabla)\bu,\tu^{n+1})+A(\bu\cdot\nabla\bar\theta,\theta^{n+1})\big].
\end{aligned}
\end{equation*}

\medskip\noindent\emph{Step 4: cancellation.} Here, we add the relations of Steps~2 and~3.
The terms multiplied by $\tilde Q^{n+1}$ cancel, as in
\cite{LinYangDong2019,Yang2021nme}.
Notably, this is the only place where the scalar is used.
Let $\tilde E^{n+1}$ denote $E^{n+1}$ with $Q^{n+1}$ replaced by
$\tilde Q^{n+1}$, the value before relaxation.
Then,
\begin{equation}\label{basic}
\begin{aligned}
\tilde E^{n+1}-E^n&+\dt\left(\nu\|\omega^{n+1}\|^2+l\|\tilde u_x^{n+1}\|_\G^2+A\kappa\|\nabla\theta^{n+1}\|^2\right)
+\tfrac{\dt^2}{3}\|\delta p\|_M^2 \le\dt(\mathbf b,\tu^{n+1})-A\dt\,T_c'(\bar u_y,\theta^{n+1}).
\end{aligned}
\end{equation}
The right-hand side is the work of the buoyancy and the exchange with the conduction profile, as in \eqref{law}.

\medskip\noindent\emph{Step 5: buoyancy.} After splitting $(\mathbf b,\tu^{n+1})=(\mathbf b,\vu^{n+1})+\frac{2\dt}{3}(\mathbf b,\nabla\phi^{n+1})$, we use three bounds:
\begin{equation}\label{three}
\|u_y^{n+1}\|\le C_P\|\omega^{n+1}\|,\qquad\|\nabla\phi^{n+1}\|\le\|\delta p\|_M,\qquad\|\bar\theta\|\le C_P\|\nabla\bar\theta\|.
\end{equation}
For the first, $u_y^{n+1}=0$ on $\G$, $\curl\vu^{n+1}=\curl\tu^{n+1}=\omega^{n+1}$ (the two differ by the gradient $\phi^{n+1}$, which goes to 0), and
\eqref{id1} gives $\|\nabla\vu^{n+1}\|=\|\omega^{n+1}\|$.
The second is \eqref{Mprop} with $\nabla\phi^{n+1}=\nabla S^{-1}\delta p$.
The third is the Poincar\'e inequality.

\begin{enumerate}[(a)]
\item \emph{No heating.} Here, $\mathbf b=G\bar\theta\ey$ and $T_c'=0$.
By \eqref{three} and Young's inequality, we have the two bounds:
\begin{align*}
\dt(\mathbf b,\vu^{n+1})&\le\dt\,GC_P\|\bar\theta\|\,\|\omega^{n+1}\|
\le\tfrac{\dt\nu}{2}\|\omega^{n+1}\|^2+\dt\,\tfrac{G^2C_P^4}{2\nu}\|\nabla\bar\theta\|^2,\\[5pt]
\tfrac{2\dt^2}{3}(\mathbf b,\nabla\phi^{n+1})&\le\tfrac{\dt^2}{3}\|\delta p\|_M^2+\tfrac{\dt^2G^2C_P^2}{3}\|\nabla\bar\theta\|^2 .
\end{align*}
Combining them gives: $$(\mathbf b,\vu^{n+1})+\frac{2\dt}{3}(\mathbf b,\nabla\phi^{n+1})\leq \frac{\dt\nu}{2}\|\omega^{n+1}\|^2+\frac{\dt^2}{3}\|\delta p\|_M^2+\dt B'\|\nabla\bar\theta\|^2\,.$$ We introduce the notation $a_k:=\|\nabla\theta^k\|^2$ for the sake of convenience.
Using the triangle inequality and the bound $(x+y)^2\le2x^2+2y^2$ gives
\begin{equation*}
\|\nabla\bar\theta\|^2=\|2\nabla\theta^n-\nabla\theta^{n-1}\|^2\le\big(2\sqrt{a_n}+\sqrt{a_{n-1}}\big)^2\le8a_n+2a_{n-1}.
\end{equation*}
These are gradients at the old levels $n$ and $n-1$, whereas the dissipation of the step controls only the new level through $A\kappa a_{n+1}$.
We therefore carry them in a reserve $\Psi^n=\dt B'(\alpha a_n+\beta a_{n-1})$, for which $\Psi^n-\Psi^{n+1}=\dt B'\big(-\alpha a_{n+1}+(\alpha-\beta)a_n+\beta a_{n-1}\big)$.
Matching $8a_n+2a_{n-1}$ requires $\beta=2$ and $\alpha-\beta=8$, so $\alpha=10$. This leaves us with:
\begin{equation*}
\dt B'(8a_n+2a_{n-1})=\Psi^n-\Psi^{n+1}+10\dt B'a_{n+1},\qquad \Psi^n=\dt B'(10a_n+2a_{n-1}).
\end{equation*}
Each $a_k=\|\nabla\theta^k\|^2$ term appears in the buoyancy bounds of the next two steps, with weights $8$ and $2$, and the reserve pays the total $10$ in advance, out of the dissipation $A\kappa a_k$ of the step at which $\theta^k$ is computed; this is the hypothesis $A\kappa\ge10B'$. (The constant $10$ is not sharp; see Remark~\ref{rem:thm}.)
Insert this into \eqref{basic}.
The first two terms are absorbed by the left-hand side, which keeps $\dt\big(\frac{\nu}{2}\|\omega^{n+1}\|^2+l\|\tilde u_x^{n+1}\|_\G^2\big)=\dt\mathcal D^{n+1}$.
With $X^n=E^n+\Psi^n$, we have:
\begin{equation*}
\tilde X^{n+1}-X^n\le-\dt\big(\mathcal D^{n+1}+(A\kappa-10B')a_{n+1}\big),
\end{equation*} where $\tilde X^{n+1}$ is $X^{n+1}$ with $Q^{n+1}$ replaced by $\tilde Q^{n+1}$.

\item \emph{Heating.} Now, $\|\mathbf b\|^2\le2G^2\big(\|\bar\theta\|^2+\|T_c-T_{\rm ref}\|^2\big)$, and
\begin{align*}
\dt(\mathbf b,\vu^{n+1})&\le\tfrac{\dt\nu}{2}\|\omega^{n+1}\|^2+\dt\,\tfrac{G^2C_P^2}{\nu}\big(\|\bar\theta\|^2+\|T_c-T_{\rm ref}\|^2\big),\\
\tfrac{2\dt^2}{3}(\mathbf b,\nabla\phi^{n+1})&\le\tfrac{\dt^2}{3}\|\delta p\|_M^2+\tfrac{2\dt^2G^2}{3}\big(\|\bar\theta\|^2+\|T_c-T_{\rm ref}\|^2\big),\\
A\dt\,|T_c'(\bar u_y,\theta^{n+1})|&\le\tfrac{A\dt\kappa}{2}\|\nabla\theta^{n+1}\|^2+\tfrac{A\dt|T_c'|^2C_P^2}{2\kappa}\|\bu\|^2 .
\end{align*}
By \eqref{En}, $\|\bar\theta\|^2\le\frac{4}{A}E^n$ and $\|\bu\|^2\le4E^n$.
Inserting into \eqref{basic} gives
\begin{equation*}
\tilde E^{n+1}-E^n+\dt\mathcal D^{n+1}\le\dt\,(KE^n+c).
\end{equation*}
\end{enumerate}

\medskip\noindent\emph{Step 6: relaxation} In both cases, the bound keeps all of the dissipation $\dt\mathcal D^{n+1}$---on the right-hand side with a minus sign in (a), and on the left-hand side in (b).

As in \cite{JiangZhangZhao2022}, the relaxation is paid for out of it: by \eqref{relax},
\begin{equation*}
E^{n+1}-\tilde E^{n+1}=\mathcal E_Q(Q^{n+1})-\mathcal E_Q(\tilde Q^{n+1})\le\eta\,\dt\,\mathcal D^{n+1}.
\end{equation*}
Adding this gives \eqref{thma}, and $E^{n+1}\le(1+K\dt)E^n+c\,\dt$.
By induction,
$E^n\le(1+K\dt)^{n-1}(E^1+c\,t_n)\le e^{Kt_n}(E^1+c\,t_n)$.
\end{proof}

\begin{rem}\label{rem:thm} Here, we present three comments on the hypotheses and the scope of the theorem:
\begin{enumerate}
\item \emph{The weights.} $A$ and $C$ enter the scheme only through \eqref{sQ} and \eqref{relax}.
Part (b) holds for any $A>0$, and the computations use $A=1$ unless stated otherwise: $A$ only weights the thermal part of the energy against the kinetic part.
The large value $A\kappa\ge10B'$ is a hypothesis of the proof of (a), not a requirement of the scheme; Section~\ref{sec:decay} shows the energy decaying with both values. 
\item \emph{What part (b) does not say.} In convection, $K$ is large ($K\approx4G^2C_P^2/(\nu A)$), so part (b) does not exclude growth of the discrete solution over times of order one, and it says nothing about the accuracy of $Q$.
Both limitations are visible in Section~\ref{sec:dt}.
\item \emph{Space discretization.} The theorem concerns the time discretization with exact spatial operators.
For the Taylor--Hood implementation of Section~\ref{sec:impl}, we do not prove it, since the discrete velocity is only approximately divergence free.
In the tests of Section~\ref{sec:decay}, and in three dimensions in Section~\ref{sec:3d}, the discrete $X^n$ decreases at every step, as part (a) states.
\end{enumerate}
\end{rem}

\subsection{What the theorem does not control}\label{sec:notcontrolled}
Theorem~\ref{thm:stab} bounds the discrete energy, but it says nothing about how close $Q$ stays to its exact value $1$; the scalar can drift away from $1$ while the energy remains bounded.
This inconsistency is a known feature of SAV schemes and is the motivation for the relaxation of \cite{JiangZhangZhao2022} and for recent corrections that keep the scalar closer to its exact value \cite{ZhangShen2026}.
The only rigorous statement we have about the accuracy of $Q$ is the following.

\begin{lem}\label{lem:contract}
At every step $|Q^{n+1}-1|\le|\tilde Q^{n+1}-1|$, and $Q^{n+1}=1$ if and only if
$\mathcal E_Q(1)-\mathcal E_Q(\tilde Q^{n+1})\le\eta\,\dt\,\mathcal D^{n+1}$. 
\end{lem}
That is, the relaxation only moves the stored value of $Q$ towards $1$, never past it and never away from it.
The lemma compares the stored and computed values of the same step; it does not bound $|Q^{n+1}-1|$ by $|Q^n-1|$.
Each step's $\tilde Q^{n+1}$ carries a new deviation $\frac{2\dt}{3C}R^{n+1}$, so the stored $Q$ can drift when the reset is only partial over many steps, as for $\eta\le10^{-4}$ in Section~\ref{sec:eta}.
\begin{proof}
By \eqref{relaxQ}, $Q^{n+1}-1=\xi(\tilde Q^{n+1}-1)$ with $\xi\in[0,1]$, which gives the inequality.
If $\tilde Q^{n+1}=1$, then $Q^{n+1}=1$, and the condition holds because its left-hand side vanishes.
Suppose now that $\tilde Q^{n+1}\neq1$.
Then, $Q^{n+1}=1$ exactly when $\xi=0$.
Because $\xi$ is the smallest admissible value in $[0,1]$, it equals $0$ exactly when $\xi=0$ satisfies \eqref{relax}.
With $\xi=0$, \eqref{relax} is the stated condition.
\end{proof}

The condition of Lemma~\ref{lem:contract} compares two quantities.
The left-hand side, $\mathcal E_Q(1)-\mathcal E_Q(\tilde Q^{n+1})$, is the change in the discrete energy $E^{n+1}$ when the computed value $\tilde Q^{n+1}$ is replaced by the exact value $1$---this is the \emph{cost} of resetting $Q$.
The right-hand side, $\eta\,\dt\,\mathcal D^{n+1}$, is the part of the dissipation of the step that \eqref{relax} allows the relaxation to spend while keeping the energy estimate of Theorem~\ref{thm:stab}---this is the \emph{budget}.
The relaxation resets $Q$ to $1$ precisely when the cost does not exceed the budget, and a reset with negative cost, which lowers the energy, is always allowed.

To compute the cost, we write $d=\tilde Q^{n+1}-1$ and $m=Q^n-1$.
With $\mathcal E_Q$ from \eqref{relax}, a direct computation gives the exact expression
\begin{equation*}
\mathcal E_Q(1)-\mathcal E_Q(\tilde Q^{n+1})=-\tfrac{3C}{2}\,d+C\,d\,m-\tfrac{5C}{4}\,d^2 .
\end{equation*}
With this in mind, two cases arise:
\begin{enumerate}
\item If $\tilde Q^{n+1}\ge1$ and $Q^n\le\frac{5}{2}$, the cost equals $-C\,d\,\big(\frac{3}{2}-m+\frac{5}{4}d\big)\le0$: lowering a value above $1$ to $1$ does not increase the energy, and the relaxation always does it.
\item If $\tilde Q^{n+1}<1$ and $Q^n=1$, the cost equals $\frac{3C}{2}|d|-\frac{5C}{4}d^2$, which is positive but at most
$\frac{3C}{2}|d|$.
Raising $Q$ to $1$ is then allowed whenever $\frac{3C}{2}|\tilde Q^{n+1}-1|\le\eta\,\dt\,\mathcal D^{n+1}$.
\end{enumerate}

Now, we compare the sizes of the cost and the budget.
Once the relaxation has reset $Q$, we have $Q^n=Q^{n-1}=1$, and \eqref{sQ} reduces to $\frac{3C}{2\dt}(\tilde Q^{n+1}-1)=R^{n+1}$, with $R^{n+1}$ from Step~3 of Section~\ref{sec:scheme}.
Therefore, the cost is at most
\begin{equation*}
\tfrac{3C}{2}\,|\tilde Q^{n+1}-1|=\dt\,|R^{n+1}|,
\end{equation*}
a quantity that does not depend on $C$.
For a smooth solution, the extrapolations $\bu$ and $\bar\theta$ differ from $\tu^{n+1}$ and $\theta^{n+1}$ by $O(\dt^2)$, so $R^{n+1}=O(\dt^2)$ and the cost is $O(\dt^3)$.
By contrast, the budget is $O(\dt)$ since $\mathcal D^{n+1}$ is of the order of the rate at which the flow dissipates kinetic energy through viscosity and wall friction.
In terms of the kinetic energy $E$ and the time-scale $\tau$ used for the weight $C$ in Section~\ref{sec:scheme}, the cost is of order $\dt\,\frac{E}{\tau}\big(\frac{\dt}{\tau}\big)^2$ and the budget of order $\eta\,\dt\,\frac{E}{\tau}$. That is, the budget exceeds the cost by a factor of order $\eta(\tau/\dt)^2$.
Consequently, for any $C$, the relaxation can reset $Q$ to $1$ at every step when $\eta \left(\tau / \dt \right)^2$ is large and the flow is not close to reset, where $\mathcal D^{n+1}$ tends to zero. This is a scaling argument though; when the budget is too small, the reset if only partial. In the runs of Section \ref{sec:eta}, the reset is exact at every step for $\eta > 0.5$. For $\eta = 10^{-2},$ it falls short brieftly at $\dt=3\cdot10^{-3}$, and for $\eta < 10^{-4},$ the reset is only partial.

Lemma~\ref{lem:contract} concerns only the stored value $Q^{n+1}$, which enters the next step through \eqref{sQ}.
The velocity and temperature of step $n+1$, on the other hand, are computed in Step~1 with $\tilde Q^{n+1}$, before the relaxation takes place.
So, the relaxation cannot change the deviation $\tilde Q^{n+1}-1=\frac{2\dt}{3C}R^{n+1}$ with which the velocity and temperature fields are computed.
That deviation is controlled only by the time-step, through $R^{n+1}=O(\dt^2)$, and by the weight $C$.
For this reason, we measure the accuracy of $Q$ by the range of $\tilde Q^n$ over a run, rather than by that of the stored values $Q^n$ (Section~\ref{sec:dt}).

\subsection{Why we do not prove an error estimate}\label{sec:noerr}
We do not pursue an error estimate here.
The obstruction is the combination of four features of the problem:
\begin{enumerate}
\item \emph{Growth constants.} Any energy-method estimate carries a Gronwall factor.
Already the bound of Theorem~\ref{thm:stab}(b) grows like $e^{Kt}$, with $K\approx4G^2C_P^2/(\nu A)$, which is of order $Ra^2$ in units of $\kappa/H^2$.
The error equation contains the same buoyancy and stratification terms, plus the linearization of the convection around the exact solution.
Part of this growth is expected; above $Ra_c$, perturbations of the conduction state grow \cite{Chandrasekhar1961}, and energy methods give decay only below the energy-stability threshold \cite{Straughan2004}.
Moreover, time-dependent convection is sensitive to perturbations, so no estimate uniform in time can hold.
This is the main difference from the settings in which such estimates are available.
In \cite{WagnerWohlmuthZawallich2026} (homogeneous temperature conditions) and \cite{AlhomsiWuZheng2026} (stable stratification), no heat enters through the walls, and error estimates are proved.
The zero-energy-contribution scheme of \cite{SunXiaZhang2025} is for the Cahn--Hilliard--Navier--Stokes system, which is dissipative, and second-order error estimates are proved for it.
The stability analysis of \cite{JiangYang2023anm} also assumes homogeneous conditions.
\item \emph{Explicit convection.} Bounding $\tilde Q^{n+1}(\bu\cdot\nabla)\bu-Q(\vu\cdot\nabla)\vu$ in the error equation requires bounds on the discrete velocity beyond those of the energy estimate (for example, in $W^{1,\infty}$), which are usually obtained by induction under a condition $\dt\le\dt_0$ that depends on the solution.
Section~\ref{sec:dt} shows that a condition of this kind, $\dt\max|\nabla\vu|\lesssim1.2$, is needed in practice, although the energy estimate holds for every $\dt$.
The same bounds would be needed to control $\tilde Q^{n+1}-1=\frac{2\dt}{3C}R^{n+1}$.
\item \emph{Pressure correction.} For the rotational BDF2 pressure correction without convection or buoyancy, the known estimates are of order $\dt^2$ for the velocity in $L^2$, but only $\dt^{3/2}$ for the velocity in $H^1$ and the pressure in $L^2$ \cite{GuermondShen2004}, \cite[Thm.~3.3]{GMS2006}.
These analyses use no-slip walls; the interaction of the rotational term with the Navier-slip condition (Lemma~\ref{lem:walls}) has not been analyzed.
For the Navier--Stokes equations with an energy-based scalar, Li, Shen, and Liu \cite{LiShenLiu2022} prove error estimates for a first-order SAV pressure-correction scheme, and energy stability but no error estimate for their second-order rotational one.
In our manufactured-solution test, we observe second-order also for $p$, so the order-$\frac{3}{2}$ bound may not be sharp in a flat channel. A pressure error estimate with Navier-slip walls remains open.
\item \emph{Relaxation.} By Lemma~\ref{lem:contract}, the relaxation never increases $|Q-1|$ within a step, but the relaxed values enter the next step through the BDF2 combination $4Q^n-Q^{n-1}$. So, a per-step contraction does not by itself transfer an estimate from the unrelaxed scheme. The relaxed schemes of \cite{JiangZhangZhao2022} are likewise analyzed for stability only.
Moreover, the relaxation does not act on $\tilde Q^{n+1}$, with which the velocity and temperature fields are computed.
\end{enumerate}
Without an error estimate, our evidence for convergence is numerical: second-order in velocity, pressure, and temperature for a manufactured solution (Table~\ref{tab:mms}), second-order convergence during the formation of the rolls (Table~\ref{tab:trans}), and time-step independence of all steady-states.
An error estimate for the first-order (backward Euler) version on bounded time intervals, under a step restriction and with constants growing exponentially in $t$, seems feasible with the techniques of \cite{GuermondShen2004,LiShenLiu2022}.
We leave it for future work.

\subsection{Three dimensions}\label{sec:stab3d}
\begin{rem}[Three dimensions]\label{rem:3d}
Lemma~\ref{lem:walls}, and with it Theorem~\ref{thm:stab}, hold without change in a channel
$[0,L_x)\times(0,H)\times[0,L_z)$, periodic in $x$ and $z$, with the Navier-slip condition
$l\vv_\tau+\nu\partial_\vn\vv_\tau=0$ on both tangential components $\vv_\tau=(v_x,v_z)$, and with
$\mathcal D^{n+1}=\frac{\nu}{2}\|\curl\tu^{n+1}\|^2+l\|\tilde\vu_\tau^{n+1}\|_\G^2$.
\begin{itemize}
\item \emph{Identity \eqref{id1}.} Here,
\begin{equation*}
|\nabla\vv|^2-|\curl\vv|^2-(\nabla\cdot\vv)^2=\nabla\cdot\big((\vv\cdot\nabla)\vv-(\nabla\cdot\vv)\vv\big).
\end{equation*}
The contributions of the periodic faces cancel.
On a wall, the normal component of the flux is
$\pm(v_x\partial_xv_y+v_z\partial_zv_y)=0$ since $v_y=0$ there and so do its tangential derivatives.
\item \emph{Identity \eqref{id2}.} With $\boldsymbol\omega=\curl\vv$,
$(\curl\curl\vv,\vv)=\|\boldsymbol\omega\|^2+\int_\G(\vn\times\boldsymbol\omega)\cdot\vv\,ds$.
On a wall
$\omega_z=-\partial_yv_x$ and $\omega_x=\partial_yv_z$, so
\begin{equation*}
(\vn\times\boldsymbol\omega)\cdot\vv=n_y(\omega_zv_x-\omega_xv_z)=-\vv_\tau\cdot\partial_\vn\vv_\tau=\tfrac{l}{\nu}|\vv_\tau|^2 .
\end{equation*}
\end{itemize}
The rest of the analysis (Poincar\'e inequality in $y$, BDF2 identities, pressure terms) does not depend on the dimension.
The three-dimensional code needs the same constant-coefficient solves per step, with one more velocity component.

What {\em does} depend on the walls being flat is Lemma~\ref{lem:walls}.
On curved walls, both identities acquire curvature terms; in two dimensions, the Navier-slip condition reads $\omega=(2\varkappa-l/\nu)\,\vv\cdot\boldsymbol\tau$, with $\varkappa$ the curvature \cite{Kelliher2006}, which would have to be absorbed by the dissipation.
\end{rem}

\section{Implementation}\label{sec:impl}
Equations \eqref{su}--\eqref{sQ} are linear in $\tilde Q^{n+1}$.
So, as is standard for SAV schemes
\cite{SXY2018,LinYangDong2019,Yang2021nme}, the unknowns for the velocity and temperature fields can be split into a part independent of $\tilde Q^{n+1}$ and
a part proportional to it.
Writing
$$\tu^{n+1}=\vu_1+\tilde Q^{n+1}\vu_2 \quad \textrm{and}\quad \theta^{n+1}=\theta_1+\tilde Q^{n+1}\theta_2,$$ we obtain the four problems:
\begin{align*}
&\left(\tfrac{3}{2}-\dt\nu\Delta\right)\vu_1=2\vu^n-\tfrac{1}{2}\vu^{n-1}-\dt\nabla p^n+\dt G(T_c+\bar\theta-T_{\rm ref})\ey,\\
& \left(\tfrac{3}{2}-\dt\nu\Delta\right)\vu_2=-\dt(\bu\cdot\nabla)\bu,\\
&\left(\tfrac{3}{2}-\dt\kappa\Delta\right)\theta_1=2\theta^n-\tfrac{1}{2}\theta^{n-1}-\dt T_c'\bar u_y,\\
&\left(\tfrac{3}{2}-\dt\kappa\Delta\right)\theta_2=-\dt\,\bu\cdot\nabla\bar\theta,
\end{align*}
with the (homogeneous) wall conditions of \eqref{su} for both $\vu_k$ and $\theta_k=0$ on $\G$.
These are four constant-coefficient Helmholtz problems (six scalar ones when counting velocity components).
With the pressure Poisson problem of Step~2, each step solves four constant-coefficient Helmholtz problems (six scalar ones) and one Poisson problem, all independent of $\tilde Q^{n+1}$.
Then, we update $\tilde Q^{n+1}$:
\begin{equation}\label{Qsol}
\tilde Q^{n+1}=\frac{2Q^n-\frac{1}{2}Q^{n-1}+\dt I_1/C}{\frac{3}{2}-\dt I_2/C},\qquad
I_k=((\bu\cdot\nabla)\bu,\vu_k)+A(\bu\cdot\nabla\bar\theta,\theta_k).
\end{equation}
The denominator is at least $\frac{3}{2}$ for every $C>0$.
Taking the inner product of $(\bu\cdot\nabla)\bu$ and the equation for $\vu_2$ provides
$$( (\bu\cdot\nabla)\bu,\vu_2)=-\dt^{-1}\left(\frac{3}{2}\|\vu_2\|^2+\dt\nu\|\nabla\vu_2\|^2+\dt l\|u_{2,x}\|_\G^2\right)\le0.$$ Likewise, $A(\bu\cdot\nabla\bar\theta,\theta_2)\le0$. So, $I_2\le0$, and every step is uniquely solvable.

For the relaxation, with $d=\tilde Q^{n+1}-1$ and $Q=1+\xi d$, the constraint in \eqref{relax} reads $a\xi^2+b\xi+c\le0$ with $a=\frac{5C}{4}d^2$, $b=\frac{Cd}{2}(5-2Q^n)$, $c=\frac{C}{4}(1+(2-Q^n)^2)-\mathcal E_Q(\tilde Q^{n+1})-\eta\dt\mathcal D^{n+1}$ (dividing by $C$, the budget becomes $\eta\dt\mathcal D^{n+1}/C$).
At $\xi=1$, the left side is $-\eta\dt\mathcal D^{n+1}\le0$; so, $\xi=0$ if $c\le0$ and otherwise, $\xi$ is the smaller root $(-b-\sqrt{b^2-4ac})/(2a)$ as in \cite{JiangZhangZhao2022}.
The budget $\mathcal D^{n+1}$ is assembled from $\tu^{n+1}$ before the projection.
For rigid walls ($l\ge10^5$), we omit the term $l\|\tilde u_x^{n+1}\|_\G^2$ from $\mathcal D^{n+1}$; a smaller budget only makes \eqref{relax} stricter, so Theorem~\ref{thm:stab} still applies.
We record the smallest and largest $\tilde Q^n$ of each run as the diagnostic of the accuracy of $Q$.
In the relaxed runs below, the stored $Q^n$ equals $1$ at every step unless stated otherwise.
Algorithm~\ref{alg:step} summarizes one step and shows where the relaxation acts.

\begin{algorithm}[htbp]
\centering
\fbox{\begin{minipage}{0.94\textwidth}\small
\textbf{Given} $\vu^n,\vu^{n-1},p^n,\theta^n,\theta^{n-1}$ and the stored $Q^n,Q^{n-1}$:
\begin{enumerate}
\item Extrapolate $\bu$ and $\bar\theta$, and solve the four constant-coefficient Helmholtz problems for
$\vu_1,\vu_2,\theta_1,\theta_2$ (independent of $\tilde Q^{n+1}$).
\item Compute $\tilde Q^{n+1}$ from \eqref{Qsol}; set $\tu^{n+1}=\vu_1+\tilde Q^{n+1}\vu_2$ and
$\theta^{n+1}=\theta_1+\tilde Q^{n+1}\theta_2$.
\emph{The fields of the step are now fixed---they use $\tilde Q^{n+1}$.}
\item Pressure correction \eqref{sp}: solve the Neumann problem for $\phi^{n+1}$, then update $\vu^{n+1}$ and
$p^{n+1}$.
\item Relaxation \eqref{relax}: compute $\mathcal D^{n+1}$ from $\tu^{n+1}$ and store $Q^{n+1}=\xi\tilde Q^{n+1}+(1-\xi)$.
\emph{Only this stored value enters the next step, through $2Q^{n+1}-\frac{1}{2}Q^n$ in \eqref{Qsol}.}
\item Record $\tilde Q^{n+1}$, and report its range over the run.
\end{enumerate}

For $Q^n=Q^{n-1}=1$, $\tilde Q^{n+1}-1=\frac{2\dt}{3C}R^{n+1}$ is fixed in step 2 by
$\dt$ and $C$; step 4 cannot change it.
Step 4 undoes, as far as the budget $\eta\dt\mathcal D^{n+1}$ allows, the accumulation of such deviations in the
stored $Q$ from one step to the next.
\end{minipage}}
\caption{One time-step of the scheme.}\label{alg:step}
\end{algorithm}

\paragraph{Spatial discretization} We use Taylor--Hood finite elements, continuous piecewise quadratic velocity and temperature and piecewise linear pressure, on a mesh of $N_x\times N_y$ rectangles, each split into two triangles, periodic in $x$.
For rigid walls, the mesh points are placed at $y_j=\frac{H}{2}(1-\cos(j\pi/N_y))$ to refine the mesh towards the walls.
%The spacing at the walls is $\frac{H}{2}(1-\cos\frac{\pi}{N_y})\approx\frac{\pi^2H}{4N_y^2}$ (this corresponds to $2.4\cdot10^{-3}H$ for $N_y=32$ and $1.1\cdot10^{-3}H$ for $N_y=48$), and $\frac{\pi H}{2N_y}$ at mid-height.

The Navier-slip condition is a natural condition, as it contributes $l\int_\G u_xv_x\,ds$ to the bilinear form of $u_x$, while $u_y=0$ and $\theta=0$ are imposed on $\G$.
Since the discrete velocity is only approximately divergence free, both nonlinear terms are written in the skew-symmetric form \cite{Temam1977}:
\begin{equation}\label{skew}
b(\mathbf w,v,z)=\tfrac{1}{2}\big((\mathbf w\cdot\nabla v,z)-(\mathbf w\cdot\nabla z,v)\big),
\end{equation}
applied componentwise to the velocity: $((\bu\cdot\nabla)\bu,\mathbf v)$ and $(\bu\cdot\nabla\bar\theta,\psi)$ in \eqref{su}--\eqref{sQ} are replaced by $b(\bu,\bu,\mathbf v)$ and $b(\bu,\bar\theta,\psi)$.
For the exact solution, the two forms coincide.
Moreover, $b(\mathbf w,z,z)=0$ for any $\mathbf w$, so the right-hand side of \eqref{sQ} is again the energy error of the explicit terms.
The cancellation of the explicit terms against the equation of $Q$ therefore holds exactly for the Galerkin scheme, and so does the bound on the denominator of \eqref{Qsol}, by testing the discrete equation for $\vu_2$ with $\vu_2$.
The rest of the proof of Theorem~\ref{thm:stab} does not carry over automatically: it uses $\nabla\cdot\vu^{n+1}=0$ and the identities of Lemma~\ref{lem:walls} for the projected velocity, which the Taylor--Hood velocity satisfies only approximately.
In Step~2, the pressure increment solves the Neumann problem $(\nabla\phi^{n+1},\nabla q)=-\frac{3}{2\dt}(\nabla\cdot\tu^{n+1},q)$ in the pressure space.
Its matrix is singular, since $\phi^{n+1}$ is determined only up to a constant \cite{BochevLehoucq2005}. So, we add a small $L^2$ term $10^{-10}(\phi^{n+1},q)$, whose solution tends to the minimal-norm (zero-mean) solution as its coefficient tends to zero \cite{KaleemGebhardtRomero2026}, and then set the mean of $\phi^{n+1}$ to zero.
Further, $\vu^{n+1}$ is the $L^2$ projection of $\tu^{n+1}-\frac{2\dt}{3}\nabla\phi^{n+1}$ onto the velocity space with $u_y^{n+1}=0$ on $\G$, and $\nabla\cdot\tu^{n+1}$ in the pressure update is its $L^2$ projection onto the pressure space.
All matrices---three Helmholtz operators, the pressure Laplacian, and three mass matrices (the velocity space with and without $u_y=0$ on $\G$, and the pressure space)---are assembled and factorized once using a sparse, direct solver (so that there are no iteration tolerances).
Each step assembles right-hand sides and solves with the stored factors.

\paragraph{Hydrostatic pressure} The computed pressure is the deviation from the hydrostatic pressure $p_c$ of the conduction state, $\nabla p_c=G(T_c-T_{\rm ref})\mathbf e_y$, which is removed analytically together with the buoyancy it balances.
This changes nothing in the equations, but matters for the discretization since Taylor--Hood elements are not pressure robust \cite{JohnLinkeMerdonNeilanRebholz2017}.
If the quadratic $p_c$ is left to the piecewise linear pressure, the unbalanced part of the buoyancy drives a steady spurious flow, of size $O(Gh^2)$, which near onset is larger than the perturbation whose growth is measured.

\paragraph{Scope of the analysis} The analysis of Section~\ref{sec:stab} is for the time discretization with exact spatial operators.
We do not
claim a fully discrete estimate, since the discrete $\vu^{n+1}$ does not satisfy $\nabla\cdot\vu^{n+1}=0$ exactly.
Instead, we verify the energy behavior of the finite-element scheme numerically (Sections~\ref{sec:decay} and~\ref{sec:3d}).

 \paragraph{Reproducibility} The computations use FreeFEM 4.13 \cite{Hecht2012} (one script, \texttt{nsb.edp}, for all two-dimensional runs, and \texttt{nsb3d.edp} for Section~\ref{sec:3d}) on one core.
The supplementary code contains the script, a list of the command line of every run reported in this work, which fixes all parameters, meshes, and time-steps, and the Python scripts that produce the tables, figures, and fitted values from the run logs.
Meshes are generated by FreeFEM's \texttt{square} with the cosine mapping for rigid walls.
Unless stated otherwise, the initial temperature is one mode, $\theta=a\sin(\pi y/H)\cos(kx)$, with the amplitude $a$ given with each test: $a=10^{-6}$ for the onset of convection, so that the flow stays linear, and $a$ between $10^{-2}$ and $0.2$ for the nonlinear runs, where it only starts the flow.
The run at $Ra=10^6$ uses independent uniform random values in $[-10^{-2},10^{-2}]$ at the nodes of the $P_2$ temperature space, multiplied by $\sin(\pi y/H)$ so that they vanish on the walls, and uses FreeFEM's generator with \texttt{randinit(1)}.
A computed state is called \emph{steady} when the Nusselt number at the bottom wall varies by less than $2\cdot10^{-8}$ relative to its value over the last quarter of the run and the two wall values agree.
There are two exceptions, also noted where they occur:
\begin{enumerate}
\item the free-slip rolls at $r=1.02$, which approach their steady-state geometrically (by a factor $0.17$ per $6$ time units at $t=60$), so that the remaining change is about $3\cdot10^{-7}$ relative difference (within $8\cdot10^{-6}$ of $Nu-1$).
\item the runs with $C=1$ and $\dt=3\cdot10^{-3}$ in Tables~\ref{tab:dt} and~\ref{tab:eta}, in which $\tilde Q^n$ still oscillates at $t=4$ and $Nu$ varies by $3.5$--$6\cdot10^{-7}$ relative difference, although its value agrees with the converged one to the digits shown.
\end{enumerate}

\section{Numerical experiments}\label{sec:num}
Unless stated otherwise, $A=1$, $\eta=0.95$, $C=1$, and the relaxation is on.
With $Q^n=1$, a steady-state of the scheme has $\tilde Q^n=1$ and is independent of $\dt$ and $C$.
The parameter $C$ matters for transient flows, where the solution changes in time before reaching a steady state, and time-dependent flows (Section~\ref{sec:dt}).

\subsection{Manufactured solution}\label{sec:mms}
On $\Omega=[0,2\pi)\times(-1,1)$ we take the exact solution:
\begin{align}\label{mms}
&\vu=\big(\partial_y\psi,-\partial_x\psi\big)=\sin t\,\big(-2y\sin x,\,-(1-y^2)\cos x\big),\\[2pt]
&\psi=\sin t\sin x\,(1-y^2),\\[2pt]
&p=\sin t\cos x\sin y,\\[2pt]
 &T=\tfrac{1-y}{2}+\tfrac{1}{2}\cos t\,(1-y^2)\big(\sin x+\tfrac{y}{2}\cos x\big),
\end{align} with $T_b=1$, $T_t=0$, $\nu=0.1$, $\kappa=0.05$, $G=1$, $T_{\rm ref}=0.5$, and $l=1$. 
The velocity is divergence free with $u_y=0$ on the walls, and $T$ takes the wall values $T_b$ and $T_t$. 
We add forcing terms that make \eqref{mms} an exact solution \cite{Roache2002}: $\mathbf f=(f_x,f_y)$ in the momentum equation, $f_T$ in the heat equation, and $g_w$ in the Navier-slip condition, obtained by substituting \eqref{mms} into each (\ref{app:mms} gives them in full, with the rest of the setup). 
We use a uniform $40\times40$ mesh. Its spatial error is fixed while $\dt$ varies, so we measure the error of the time discretization against a run with a much smaller $\dt$ on the same mesh. 
Table~\ref{tab:mms} and Figure~\ref{fig:conv}(a) show observed second-order convergence in velocity, pressure, and temperature up to $t=0.5$ for the rotational scheme.
The standard incremental form (the rotational term dropped from \eqref{sp}) imposes $\partial_\vn p^{n+1}=\partial_\vn p^n$ on the walls, which creates a numerical boundary layer of width $O(\sqrt{\nu\dt})$ \cite{GMS2006}.

On the uniform mesh (element height $0.05$), this layer of width $\sqrt{\nu\dt}\le0.035$ for $\dt\le0.0125$ is thinner than the elements, and the standard form is also second-order, with a pressure error about $1.2$ times that of the rotational form (Table~\ref{tab:mms}).
To resolve the layer, we repeat the test on a mesh graded towards the walls, with spacing $3\cdot10^{-3}$ at the walls, below $\sqrt{\nu\dt}$ for all steps used.
There, the pressures of the standard and rotational forms differ by $1.6\cdot10^{-4}$, $2.2\cdot10^{-5}$, and $1.6\cdot10^{-6}$ at $\dt=0.0125$, $0.003125$, and $0.00078125$, which corresponds to orders $1.45$ and $1.88$ between successive steps.
At the two smaller steps, this difference is three to four times the time error of the rotational form on the uniform mesh (Table~\ref{tab:mms}, extrapolated at second-order to $\dt=0.00078125$).
Once the layer is resolved, the standard form is measurably less accurate with the pressure, as the theory predicts \cite{GuermondShen2004,GMS2006}. In these tests, we note its observed order exceeds the first-order bound of that theory.

\emph{Spatial convergence.} Table~\ref{tab:space} gives the errors against the exact solution on uniform $N\times N$ meshes, with $\dt=2.5\cdot10^{-4}$.
For $P_2$--$P_1$ elements and a smooth solution, we should expect order $3$ for the velocity
and the temperature in $L^2$, and order $2$ for their gradients and for the pressure in $L^2$ \cite{GiraultRaviart1986}.
The observed orders approach these values, $3.04$, $3.00$, $2.02$, $2.00$, and $2.01$ between the two finest meshes, and $\|\nabla\cdot\vu_h\|$, which is not zero for Taylor--Hood elements, decreases at order $2$.
The time error does not affect the finest mesh; rerunning $N=64$ with $\dt=5\cdot10^{-4}$ changes the errors by at most $0.12\%$, as expected from Table~\ref{tab:mms}, which puts the time error at $\dt=2.5\cdot10^{-4}$ below $10^{-7}$.

\begin{table}[htbp]
\centering
\setlength{\tabcolsep}{2pt}\footnotesize
\begin{tabular}{r cc cc cc c c}
\toprule
& \multicolumn{6}{c}{rotational, relaxed} & standard & unrelaxed\\
\cmidrule(lr){2-7}\cmidrule(lr){8-8}\cmidrule(lr){9-9}
$\dt$ & $\vu$ & order & $p$ & order & $T$ & order & $p$ & $|Q^N-1|$\\
\midrule
0.05 & $2.40{\times}10^{-3}$ & -- & $1.68{\times}10^{-3}$ & -- & $3.56{\times}10^{-3}$ & -- & $1.70{\times}10^{-3}$ & $3.54{\times}10^{-6}$\\
0.025 & $6.02{\times}10^{-4}$ & 1.99 & $4.25{\times}10^{-4}$ & 1.99 & $8.83{\times}10^{-4}$ & 2.01 & $5.03{\times}10^{-4}$ & $2.01{\times}10^{-7}$\\
0.0125 & $1.51{\times}10^{-4}$ & 2.00 & $1.07{\times}10^{-4}$ & 1.99 & $2.20{\times}10^{-4}$ & 2.01 & $1.27{\times}10^{-4}$ & $1.23{\times}10^{-8}$\\
0.00625 & $3.78{\times}10^{-5}$ & 2.00 & $2.67{\times}10^{-5}$ & 2.00 & $5.47{\times}10^{-5}$ & 2.00 & $3.16{\times}10^{-5}$ & $7.70{\times}10^{-10}$\\
0.003125 & $9.44{\times}10^{-6}$ & 2.00 & $6.66{\times}10^{-6}$ & 2.00 & $1.36{\times}10^{-5}$ & 2.01 & $7.90{\times}10^{-6}$ & $4.75{\times}10^{-11}$\\
\bottomrule
\end{tabular}
\caption{Manufactured solution (Section~\ref{sec:mms}), Taylor--Hood elements on a $40\times40$ mesh: $L^2$ differences at
$t=0.5$ from a run with $\dt=1.953125\cdot10^{-4}$ on the same mesh, which isolate the error of the time
discretization, and observed orders, for the rotational scheme with relaxation. The next column is the pressure of the
standard incremental form, measured in the same way; the last is $|Q^N-1|$ without relaxation.}
\label{tab:mms}
\end{table}

\begin{table}[htbp]\centering\small\setlength{\tabcolsep}{3pt}
\resizebox{\textwidth}{!}{\begin{tabular}{rllllllllllll}\toprule
$N$ & \multicolumn{2}{c}{$\|\vu-\vu_h\|$} & \multicolumn{2}{c}{$|\vu-\vu_h|_{1}$} & \multicolumn{2}{c}{$\|\nabla\cdot\vu_h\|$} & \multicolumn{2}{c}{$\|T-T_h\|$} & \multicolumn{2}{c}{$|T-T_h|_{1}$} & \multicolumn{2}{c}{$\|p-p_h\|$}\\
& error & order & error & order & error & order & error & order & error & order & error & order\\\midrule
8 & $8.36\cdot10^{-3}$ &  & $1.20\cdot10^{-1}$ &  & $9.30\cdot10^{-2}$ &  & $4.01\cdot10^{-3}$ &  & $7.48\cdot10^{-2}$ &  & $2.11\cdot10^{-2}$ & \\
16 & $9.12\cdot10^{-4}$ & 3.20 & $2.72\cdot10^{-2}$ & 2.14 & $2.00\cdot10^{-2}$ & 2.22 & $5.27\cdot10^{-4}$ & 2.93 & $1.91\cdot10^{-2}$ & 1.97 & $4.97\cdot10^{-3}$ & 2.09\\
32 & $1.06\cdot10^{-4}$ & 3.11 & $6.47\cdot10^{-3}$ & 2.07 & $4.55\cdot10^{-3}$ & 2.13 & $6.62\cdot10^{-5}$ & 2.99 & $4.79\cdot10^{-3}$ & 2.00 & $1.22\cdot10^{-3}$ & 2.02\\
64 & $1.28\cdot10^{-5}$ & 3.04 & $1.59\cdot10^{-3}$ & 2.02 & $1.10\cdot10^{-3}$ & 2.05 & $8.27\cdot10^{-6}$ & 3.00 & $1.20\cdot10^{-3}$ & 2.00 & $3.05\cdot10^{-4}$ & 2.01\\
\midrule expected &  & 3 &  & 2 &  & 2 &  & 3 &  & 2 &  & 2\\
\bottomrule\end{tabular}}
\caption{Spatial convergence of the Taylor--Hood discretization for the manufactured solution (\ref{app:mms}): errors at $t=0.5$ on uniform $N\times N$ meshes with $\dt=2.5\cdot10^{-4}$, against the exact functions; $|\cdot|_1$ is the $H^1$ seminorm, and the pressures are compared after removing their means. Orders from successive meshes. The expected orders are those of $P_2$--$P_1$ elements for smooth solutions.}
\label{tab:space}
\end{table}

\begin{figure}[htbp]
\centering
\includegraphics[width=\textwidth]{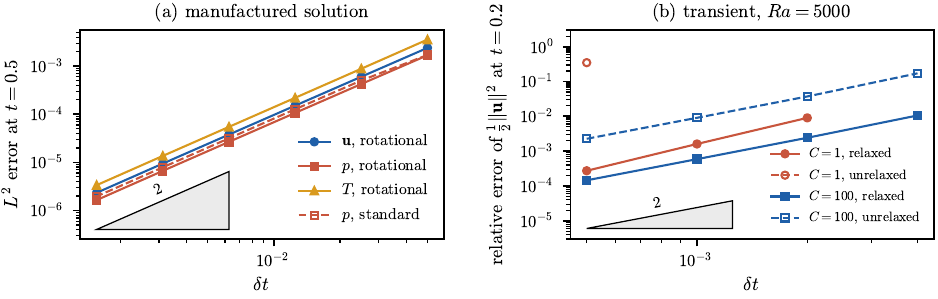}
\caption{Convergence in time. (a) Manufactured solution (Section~\ref{sec:mms}): $L^2$ errors at $t=0.5$ of the rotational scheme with relaxation, and the pressure error of the standard incremental form (dashed), which on this uniform mesh is also second-order. (b) The transient period at $Ra=5000$ (Section~\ref{sec:dt}): relative error of $\frac{1}{2}\|\vu\|^2$ at $t=0.2$ against a computation with $\dt=3.125\cdot10^{-5}$, for the weights $C=1$ and $C=100$ of the $Q$ equation, with and without relaxation.
Points with errors above $1$ (blow-up) are omitted.
The triangles show slope $2$.}
\label{fig:conv}
\end{figure}

\subsection{Onset of convection}\label{sec:onset}
For a cell of height $H=1$ with $\nu=\kappa=1$, $T_b=1$, $T_t=0$, $G=Ra$, we perturb the conduction state by $\theta=10^{-6}\sin(\pi y)\cos(kx)$ in a cell of one wavelength $L=2\pi/k$ and measure the growth rate of the kinetic energy.
For free-slip walls ($l=0$), the linear problem is solved exactly by $\sigma=k\sqrt{Ra}/K-K^2$, $K^2=k^2+\pi^2$ \cite{Chandrasekhar1961, Rayleigh1916}.
The wave number $k=\pi/\sqrt2$ becomes unstable first, at $Ra_c=27\pi^4/4\approx657.511$.
Table~\ref{tab:growth} and Figure~\ref{fig:onset} show second-order convergence of the computed growth rates in $\dt$.

On the $32\times32$ mesh, the spatial error of $\sigma$ is about $6\cdot10^{-5}$, which dominates near onset where the time error is smallest.
There, the successive differences $\sigma(\dt)-\sigma(\dt/2)$ of Table~\ref{tab:growth} are below $10^{-5}$, so the orders computed from their ratios are less reliable.
The spatial error of this eigenvalue decreases like $h^4$ (to about $10^{-5}$ on the $48\times48$ mesh) and the zero of a quadratic fit of $\sigma(Ra)$ on the $48\times48$ mesh with $\dt=0.005$ gives $Ra_c=\RacFS$.

For rigid walls, we take $l/\nu=10^6$, a finite-slip approximation to no slip, with $k=3.1163$ and the mesh graded towards the walls.
With $\dt=0.01$, the same procedure gives $Ra_c=1707.863$, $1707.789$, and $\RacRG$ on the $24^2$, $32^2$, and $48^2$ meshes, against the no-slip value $1707.76$ \cite{Chandrasekhar1961}.
The fits use $Ra=1650$, $1700$, $1750$, and $1800$.
On the $48^2$ mesh, we add $Ra=1710$, just above onset, which lies on the fit of the other four to within $2\cdot10^{-5}$ in $\sigma$ and leaves the fitted $Ra_c$ unchanged to the digits shown.
The error of this finite-slip approximation is below the spatial error.
On the $32\times32$ mesh, the same procedure gives $Ra_c=1609.88$ (extrapolated slightly below the fitted range), $1697.00$, $1706.71$, and $1707.69$ for $l/\nu=10^2,\dots,10^5$, and the deficits from the value at $10^6$ ($97.9$, $10.8$, $1.08$, $0.098$) are approximately proportional to $\nu/l$. So, the slip lowers $Ra_c$ by about $0.01$ at $l/\nu=10^6$, about $6\cdot10^{-6}$ relative difference and below the spatial error of the mesh.

\begin{table}[htbp]
\centering
\setlength{\tabcolsep}{3pt}\footnotesize
\begin{tabular}{r r cccc c}
\toprule
$Ra$ & exact $\sigma$ & $\dt=0.04$ & $0.02$ & $0.01$ & $0.005$ & orders\\
\midrule
500 & $-1.89446$ & $-9.47{\times}10^{-2}$ & $-2.12{\times}10^{-2}$ & $-5.16{\times}10^{-3}$ & $-1.32{\times}10^{-3}$ & 2.20, 2.06\\
600 & $-0.66227$ & $-1.14{\times}10^{-2}$ & $-2.83{\times}10^{-3}$ & $-7.46{\times}10^{-4}$ & $-2.30{\times}10^{-4}$ & 2.05, 2.01\\
650 & $-0.08481$ & $-3.20{\times}10^{-4}$ & $-1.26{\times}10^{-4}$ & $-7.53{\times}10^{-5}$ & $-6.20{\times}10^{-5}$ & 1.92, 1.95\\
660 & $0.02799$ & $-4.55{\times}10^{-5}$ & $-5.31{\times}10^{-5}$ & $-5.60{\times}10^{-5}$ & $-5.69{\times}10^{-5}$ & 1.36, 1.70\\
700 & $0.47085$ & $-4.69{\times}10^{-3}$ & $-1.22{\times}10^{-3}$ & $-3.47{\times}10^{-4}$ & $-1.29{\times}10^{-4}$ & 1.99, 2.00\\
800 & $1.52553$ & $-5.02{\times}10^{-2}$ & $-1.35{\times}10^{-2}$ & $-3.51{\times}10^{-3}$ & $-9.26{\times}10^{-4}$ & 1.87, 1.96\\
1000 & $3.45301$ & $-2.44{\times}10^{-1}$ & $-7.17{\times}10^{-2}$ & $-1.91{\times}10^{-2}$ & $-4.93{\times}10^{-3}$ & 1.71, 1.89\\
\bottomrule
\end{tabular}
\caption{Onset with free-slip walls ($l=0$, $k=\pi/\sqrt2$, Taylor--Hood elements on a $32\times32$ mesh): error of the
computed growth rate, $\frac{1}{2}\frac{d}{dt}\ln\|\vu\|^2$ fitted on $t\in[1,2]$, against the exact value $\sigma=k\sqrt{Ra}/K-K^2$,
and orders in $\dt$ from successive differences $\sigma(\dt)-\sigma(\dt/2)$, which cancel the spatial error.}
\label{tab:growth}
\end{table}

\begin{figure}[htbp]
\centering
\includegraphics[width=\textwidth]{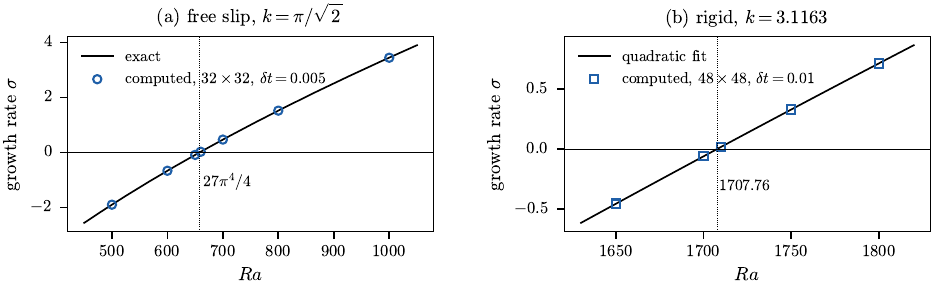}
\caption{At the onset of convection (Section~\ref{sec:onset}).
We plot the growth rate of the kinetic energy against $Ra$. (a) Free-slip walls, $k=\pi/\sqrt2$, uniform $32\times32$ mesh, $\dt=0.005$, and the exact growth rate $\sigma=k\sqrt{Ra}/K-K^2$. (b) Rigid walls, $k=3.1163$, $48\times48$ mesh graded towards the walls, $\dt=0.01$, and the quadratic fit whose zero gives $Ra_c=\RacRG$.}
\label{fig:onset}
\end{figure}

\subsection{Steady rolls}\label{sec:rolls}
The onset tests of this section use perturbations of amplitude $10^{-6}$, for which the convection and advection terms are negligible and the scalar $Q$ plays no role.
Steady rolls just above onset are the simplest flows in which these terms matter---the nonlinear terms balance the buoyancy, and $\tilde Q^{n+1}$ carries their energy error at every step.

As such, investigating steady rolls is a natural next step for several reasons.
First, steady rolls have a nice analytical reference to compare against; close to the onset of convection, weakly nonlinear theory predicts the heat transport.
Second, a steady-state is where a drift of $Q$ does the most harm; without relaxation, the scheme converges to the steady-state of a problem with rescaled convection (Section~\ref{sec:dt}).
With $Q=1$ at every step, by contrast, a steady-state of the scheme does not depend on $\dt$, so these runs measure the spatial discretization and the consistency of $Q$, separately from the time error.
Third, the rolls computed here are the reference flows for the rest of the section: the rigid-wall rolls at $Ra=5000$ for the time-step study of Section~\ref{sec:dt}, and the free-slip rolls at $r=1.5$ and $2$ for the faster flows discussed there.
This section also describes how we measure the Nusselt number.

Near onset, steady free-slip rolls transport heat as
\begin{equation*}
Nu-1\approx2\epsilon,\qquad \epsilon=\frac{Ra-Ra_c}{Ra_c},
\end{equation*}
independently of the Prandtl number \cite{MalkusVeronis1958}.
The three-mode truncation of Lorenz \cite{Lorenz1963} gives a closed form for the whole branch: with $r=Ra/Ra_c$, its steady-state and temperature expansion give
\begin{equation*}
Nu-1=\frac{2(r-1)}{r}=2\epsilon-2\epsilon^2+O(\epsilon^3),
\end{equation*}
which agrees with the weakly nonlinear result to first order in $\epsilon$.
In a steady-state, the heat flux through every horizontal line is the same, so the Nusselt number can be measured
either at a wall or in the volume:
\begin{equation*}
Nu=\frac{1}{L}\int_0^L\frac{\partial_yT}{T_c'}\,dx\ \ \text{on either wall}\qquad\text{or}\qquad
Nu=1+\frac{\langle u_y\theta\rangle}{\kappa|T_c'|}.
\end{equation*}
With finite elements, these two definitions are not equally accurate though.
The wall value inherits the $O(h^2)$ error of the gradient of the discrete temperature, while the volume average converges like the energy, here at $O(h^4)$.
On the $24\times24$ mesh, they differ by $0.5\%$ of $Nu-1$.
Therefore, we report the volume average throughout.
This is a standard definition in simulations of Rayleigh--B\'enard convection, where the agreement between the wall, volume, and dissipation-based values is also used to check the resolution \cite{StevensVerziccoLohse2010,Kooij2018}.
\begin{itemize}
\item \emph{Free-slip walls.} We start from $\theta=10^{-2}\sin(\pi y)\cos(kx)$ and integrate to $t=60$, with $\dt=0.02$ for $r\le1.2$ and $\dt=0.01$, $C=100$ for $r=1.5$ and $2$.
Table~\ref{tab:rolls} (left) lists the steady-states.
The ratio $(Nu-1)/\epsilon$ decreases from $1.967$ at $\epsilon=0.02$ to $1.920$ at $\epsilon=0.05$ and $1.848$ at $\epsilon=0.1$.
This extrapolates to $2.00$ as $\epsilon\to0$ and departs from the Lorenz value further from onset.
The values on the $32\times32$ mesh are converged in space to within $1.5\cdot10^{-5}$, by $h^4$ extrapolation from the $24\times24$ and $32\times32$ meshes.
\item \emph{Rigid walls.} For $k=3.117$ and $Pr=1$, we start from $\theta=0.1\sin(\pi y)\cos(kx)$ and integrate to $t=4$ with $C=100$ and $\dt=2\cdot10^{-3}$ ($10^{-3}$ at $Ra=10^4$).
Table~\ref{tab:rolls} (right) lists $Nu$ and the kinetic energy of the steady rolls, and Figure~\ref{fig:flows}(a) shows those at $Ra=5000$.
To check the spatial resolution, we repeated the run at $Ra=5000$ on a $48\times48$ mesh.
Going from the $32\times32$ to the $48\times48$ mesh changes $Nu$ from $2.110484$ to $2.110491$ and the kinetic energy from $142.3465$ to $142.3491$, a relative difference of $3\cdot10^{-6}$ and $2\cdot10^{-5}$ between their respective values.
The values of Table~\ref{tab:rolls} (right) are accurate to at least four significant digits in the kinetic energy and five in $Nu$.
\item \emph{Failed attempts.} Our first runs at $r=1.5$ and $2$ used $\dt=0.02$ with $C=1$, and they illustrate Section~\ref{sec:dt}.
At $r=2$, $\tilde Q^n$ collapsed to $6\cdot10^{-9}$ and the solution blew up at $t=2.4$.
At $r=1.5$, the run did not blow up.
Instead, it settled into a periodic cycle, with $\tilde Q^n$ oscillating between $0.63$ and $0.91$ and $Nu$ between $1.4$ and $2.1$, rather than reaching the steady rolls with $Nu=1.717$.
Since the stored $Q^n$ was $1$ at every step, only the range of $\tilde Q^n$ shows that this cycle is not a physical time-dependent flow.
It is produced by the scheme when $\dt$ is too large for $C=1$: with $\dt=0.01$ and $C=100$, the same run reaches the steady rolls.
\end{itemize}

\begin{table}[htbp]
\centering
\footnotesize
\begin{tabular}{c c c c @{\hspace{2em}} c c c}
\toprule
\multicolumn{4}{c}{free slip} & \multicolumn{3}{c}{rigid}\\
\cmidrule(lr){1-4}\cmidrule(lr){5-7}
$r$ & $Nu$ & $(Nu-1)/\epsilon$ & $2(r-1)/r\ /\ \epsilon$ & $Ra$ & $Nu$ & $\frac{1}{2}\|\vu\|^2$\\
\midrule
$1.02$ & $1.039330$ & $1.9665$ & $1.9608$ & $3000$ & $1.662664$ & $51.6841$\\
$1.05$ & $1.096006$ & $1.9201$ & $1.9048$ & $5000$ & $2.110484$ & $142.3465$\\
$1.1$ & $1.184755$ & $1.8476$ & $1.8182$ & $10000$ & $2.646399$ & $405.8035$\\
$1.2$ & $1.343910$ & $1.7195$ & $1.6667$ &  & & \\
$1.5$ & $1.716901$ & $1.4338$ & $1.3333$ &  & & \\
$2$ & $2.141476$ & $1.1415$ & $1.0000$ &  & & \\
\bottomrule
\end{tabular}
\caption{Steady rolls, $Pr=1$, Taylor--Hood elements on a $32\times32$ mesh, $Nu$ from the volume average
$1+\langle u_y\theta\rangle/(\kappa|T_c'|)$. Left: free-slip walls, $k=\pi/\sqrt2$, $Ra=rRa_c$, uniform mesh; the ratio
$(Nu-1)/\epsilon$ against the Lorenz value $2(r-1)/r/\epsilon$. Right: rigid walls, $k=3.117$, mesh graded towards the walls;
$Nu$ and the kinetic energy.}
\label{tab:rolls}
\end{table}

\begin{figure}[htbp]
\centering
\includegraphics[width=\textwidth]{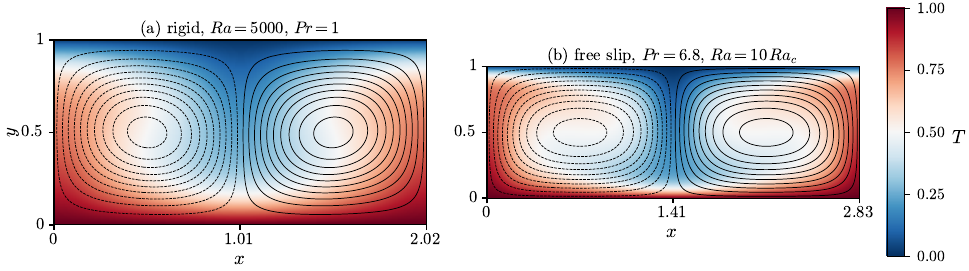}
\caption{Temperature and streamlines (contours of the stream function; dashed: clockwise) over one wavelength. (a) Steady rolls, rigid walls, $Ra=5000$, $Pr=1$, $k=3.117$ (Section~\ref{sec:dt}), $32\times32$ mesh. (b) Steady free-slip cells, $Pr=6.8$, $Ra=10\,Ra_c$ (Table~\ref{tab:mw}), $48\times32$ mesh.}
\label{fig:flows}
\end{figure}

\subsection{Comparison against classic results}\label{sec:bench}
For these comparisons, we set $\nu=Pr$, $\kappa=1$, and $G=Ra\,Pr$, so that time is measured in units of the thermal diffusion time, take one wavelength $L=2\pi/\alpha$ of rolls of wave number $\alpha$, and run with relaxation from $\theta=0.1\sin(\pi y)\cos(\alpha x)$ until the state is steady in the sense of Section~\ref{sec:impl} (variation of $Nu$ below $2\cdot10^{-8}$ relative difference over the last quarter of the run).
Steady-states do not depend on $C$ or $\dt$; we take $C=100$ (or $10^3$ at $Ra=10^4$ and for the free-slip cells) and $\dt$ from $5\cdot10^{-3}$ to $10^{-3}$, chosen for stability only.
As in Section~\ref{sec:rolls}, $Nu$ is the volume average.

\begin{itemize}
\item \emph{Rigid walls.} Clever and Busse \cite{CleverBusse1974} computed steady rolls with a Galerkin method for $Pr=0.71$ and $7$ and several wave numbers.
They estimate their Nusselt numbers to be within $\frac{1}{2}$--$1\%$ of the exact values, and note that as the truncation is refined, the heat transport first overshoots.
The overshoot is about $2$--$3\%$ at moderate Prandtl number, and then decreases to its limit as truncation error is further refined.
Table~\ref{tab:cb} and Figure~\ref{fig:bench}(a) show that, for $\alpha=3.117$ and $Ra\le10^4$, the present values are below theirs in every case, by $0.09$--$0.42\%$, consistent with that estimate.
These values are mesh independent at the precision shown: on the $48\times48$ mesh the two cases at $Ra=10^4$ (the most demanding ones) change by $1.5\cdot10^{-5}$ ($Pr=0.71$) and $1.5\cdot10^{-6}$ ($Pr=7$).

\begin{table}[htbp]
\centering
\footnotesize
\begin{tabular}{r ccr ccr}
\toprule
& \multicolumn{3}{c}{$Pr=0.71$} & \multicolumn{3}{c}{$Pr=7$}\\
\cmidrule(lr){2-4}\cmidrule(lr){5-7}
$Ra$ & present & CB & diff. & present & CB & diff.\\
\midrule
2000 & 1.2105 & 1.212 & $-0.12\%$ & 1.2130 & 1.214 & $-0.09\%$\\
3000 & 1.6602 & 1.663 & $-0.17\%$ & 1.6626 & 1.667 & $-0.27\%$\\
5000 & 2.1103 & 2.116 & $-0.27\%$ & 2.1031 & 2.112 & $-0.42\%$\\
10000 & 2.6525 & 2.661 & $-0.32\%$ & 2.6080 & 2.618 & $-0.38\%$\\
\bottomrule
\end{tabular}
\caption{Rigid walls, rolls of wave number $\alpha=3.117$: Nusselt number against Table~1 of Clever and Busse
\cite{CleverBusse1974} (CB), and the relative difference. Present values are steady states reached from
$\theta=0.1\sin(\pi y)\cos(\alpha x)$ with relaxation, on a $32\times32$ mesh graded towards the walls; $Nu$ is the volume
average of Section~\ref{sec:rolls}.}
\label{tab:cb}
\end{table}

\item \emph{Free-slip walls.} Moore and Weiss \cite{MooreWeiss1973} list, for $Pr=6.8$ and cells of width $\sqrt2$ (so $k=\pi/\sqrt2$, as in Section~\ref{sec:rolls}), the Nusselt numbers of Veronis, from a truncated Fourier expansion with errors below $1\%$, and their own finite-difference values on two meshes.
Table~\ref{tab:mw} and Figure~\ref{fig:bench}(b) show that the present values agree with those of Veronis to within the rounding of the published values up to $r=10$.
Figure~\ref{fig:flows}(b) shows the cells at $r=10$.
\end{itemize}

We have restricted the comparison to steady rolls that are stable and reached by time-stepping.
Larger Rayleigh numbers, other wave numbers, and time-dependent flows would need finer meshes and smaller time-steps with the present finite-element code.
The Ra dependence of the weight $C$ is illustrated in Sections~\ref{sec:dt} and~\ref{sec:ra1e6}.

\begin{table}[htbp]
\centering
\footnotesize
\begin{tabular}{r c c c c}
\toprule
& & \multicolumn{3}{c}{Moore and Weiss, Table 1}\\
\cmidrule(lr){3-5}
$r$ & present & Veronis & $24\times18$ & $48\times34$\\
\midrule
1.1 & 1.1848 & 1.18 & 1.16 & --\\
1.2 & 1.3440 & 1.34 & 1.32 & --\\
1.4 & 1.6067 & 1.61 & 1.58 & --\\
2 & 2.1424 & 2.14 & 2.12 & 2.14\\
3 & 2.6784 & 2.68 & 2.65 & --\\
4 & 3.0402 & 3.04 & 3.01 & --\\
6 & 3.5535 & 3.55 & 3.52 & --\\
10 & 4.2435 & 4.24 & 4.19 & --\\
\bottomrule
\end{tabular}
\caption{Free-slip walls, $Pr=6.8$, cells of width $\sqrt2$ ($k=\pi/\sqrt2$), $Ra=rRa_c$: Nusselt number
against Table~1 of Moore and Weiss \cite{MooreWeiss1973}, which lists the values of Veronis and their
finite-difference values on two meshes. Present: volume average of $Nu$, $32\times32$ mesh ($48\times32$ for $r=10$).}
\label{tab:mw}
\end{table}

\begin{figure}[htbp]
\centering
\includegraphics[width=\textwidth]{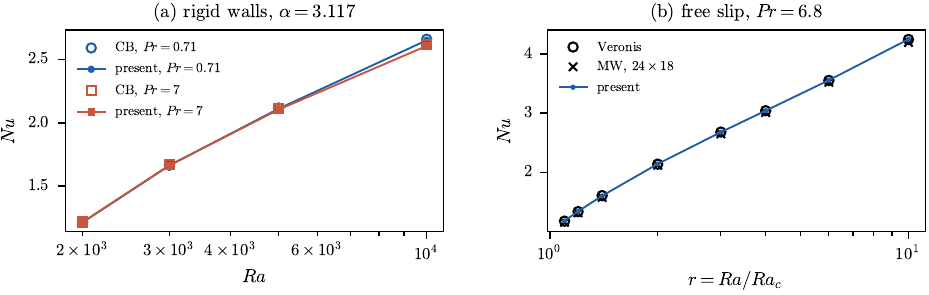}
\caption{Nusselt numbers of steady rolls against published values. (a) Rigid walls, $\alpha=3.117$, against Clever and Busse (CB) \cite{CleverBusse1974}, Table~\ref{tab:cb}. (b) Free-slip walls, $Pr=6.8$, against the values of Veronis and the coarse mesh of Moore and Weiss (MW) \cite{MooreWeiss1973}, Table~\ref{tab:mw}.} \label{fig:bench}
\end{figure}

\subsection{Accuracy of the zero-energy variable}\label{sec:dt}
We follow the formation of the rigid-wall steady rolls at $Ra=5000$ ($32\times32$ mesh) from $\theta=0.1\sin(\pi y)\cos(kx)$ to $t=4$ (Table~\ref{tab:dt}) with two weights: $C=1$, and $C=100$, which is comparable to the kinetic energy $\frac{1}{2}\|\vu\|^2\approx142$ of the rolls.

\begin{itemize}
\item \emph{$C=1$ without relaxation.} During the transient period of the flow, the perturbation grows, saturates near $t=0.2$, and then settles into the rolls (Figure~\ref{fig:C}(a)).
In this period, the flow changes quickly and the energy error $R^{n+1}$ absorbed by $Q$ is large.
This causes $Q^n$ to drift away from $1$, and once the rolls are steady, $Q^n$ stays where it is.
The computed steady-state is then that of a different problem, with convection and advection scaled by $Q$, and its Nusselt number is off by $7\%$ even at $\dt=5\cdot10^{-4}$.
For $\dt\ge10^{-3}$, $Q^n$ collapses to zero.
This switches off the advection that saturates the instability, and the perturbation grows without bound.
Theorem~\ref{thm:stab}(b) allows this, since its constant $K\approx\frac{4Ra^2}{\pi^2}$ is about $10^7$ here.

\item \emph{$C=1$ with relaxation.} The stored $Q^n$ is $1$ at every step, and the steady-state is reproduced to a relative $10^{-8}$ (all digits of Table~\ref{tab:dt}) up to $\dt=3\cdot10^{-3}$.
The value $\tilde Q^n$, with which the velocity and temperature fields are computed, is not $1$, however.
Its deviation grows with the time-step, to about $20\%$ at $\dt=3\cdot10^{-3}$, and at $\dt=4\cdot10^{-3}$ it collapses within a step and the solution blows up, while the stored $Q^n$ is still $1$ (Figure~\ref{fig:C}(c,d)).
Notably, this limit does not depend on the mesh.
On the $24\times24$, $32\times32$, and $48\times48$ meshes, the runs blow up in the same
step.
This is the amplification by the kinetic energy discussed with the weight $C$ in Section~\ref{sec:scheme}.

\item \emph{$C=100$.} In this case, the deviation of $\tilde Q^n$ from $1$ is $100$ times smaller and grows like $\dt^3$, as predicted in Section~\ref{sec:scheme}.
(With $C=10^4$, it is smaller by another factor $100$---see Figure~\ref{fig:C}(b).)
The relaxed scheme now reproduces the steady-state to a relative difference of $10^{-8}$ up to $\dt=6\cdot10^{-3}$, twice the limit for $C=1$.
Without relaxation, the drift of $Q$ is also $100$ times smaller, and the runs no longer blow up for $\dt\le6\cdot10^{-3}$.
However, the steady-states are still biased (by $1.8\%$ at $\dt=3\cdot10^{-3}$).
With $C=100$, the relaxation is therefore needed for accuracy rather than for stability.

\item \emph{The explicit convection.} At $\dt=8\cdot10^{-3}$, the steady rolls give way to a growing oscillation for $C=100$ and $C=10^4$ alike, even though $\tilde Q^n$ stays close to $1$.
Both cases blow up at $\dt=1.2\cdot10^{-2}$, suggesting this limit is set by the flow.
Beyond this limit, a larger time-step does not necessarily fail sooner.
With $C=100$ at $\dt=1.6\cdot10^{-2}$, the run reaches $t=4$, but its kinetic energy swings between less than $1$ and about $450$, against $142$ for the steady rolls, and $\tilde Q^n$ between $0.57$ and $1.52$.
The largest velocity gradient of the rolls, in the wall shear layers (Figure~\ref{fig:flows}(a)), is about $200$.
So, the limit corresponds to $\dt\max|\nabla\vu|\approx1.2$--$1.6$.
It is a restriction of the explicit treatment of the convection, which neither the energy estimate, the relaxation, nor the weight $C$ removes.
With $C=1$, the collapse of $\tilde Q$ comes first, at half this time-step.

\begin{table}[htbp]
\centering
\footnotesize
\setlength{\tabcolsep}{2pt}
\begin{tabular*}{\textwidth}{@{\extracolsep{\fill}}r cc cc cc cc@{}}
\toprule
& \multicolumn{4}{c}{$C=1$} & \multicolumn{4}{c}{$C=100$}\\
\cmidrule(lr){2-5}\cmidrule(lr){6-9}
& \multicolumn{2}{c}{relaxed} & \multicolumn{2}{c}{unrelaxed} & \multicolumn{2}{c}{relaxed} & \multicolumn{2}{c}{unrelaxed}\\
\cmidrule(lr){2-3}\cmidrule(lr){4-5}\cmidrule(lr){6-7}\cmidrule(lr){8-9}
$\dt$ & $Nu(4)$ & range of $\tilde Q^n$ & $Nu(4)$ & range of $\tilde Q^n$ & $Nu(4)$ & range of $\tilde Q^n$ & $Nu(4)$ & range of $\tilde Q^n$\\
\midrule
$1.6\cdot10^{-2}$ & blow-up & collapse & blow-up & collapse & unsteady$^\dagger$ & $0.5725$--$1.5191$ & blow-up & collapse\\
$1.2\cdot10^{-2}$ & blow-up & collapse & blow-up & collapse & blow-up & collapse & blow-up & collapse\\
$8\cdot10^{-3}$ & blow-up & collapse & blow-up & collapse & unsteady$^\dagger$ & $0.9553$--$1.0149$ & blow-up & collapse\\
$6\cdot10^{-3}$ & blow-up & collapse & blow-up & collapse & $2.110511$ & $0.9883$--$1.0044$ & $2.361453$ & $0.8093$--$1.0000$\\
$4\cdot10^{-3}$ & blow-up & collapse & blow-up & collapse & $2.110511$ & $0.9967$--$1.0010$ & $2.182896$ & $0.9274$--$1.0000$\\
$3\cdot10^{-3}$ & $2.110511$ & $0.8216$--$1.1946$ & blow-up & collapse & $2.110511$ & $0.9986$--$1.0004$ & $2.147909$ & $0.9603$--$1.0000$\\
$2\cdot10^{-3}$ & $2.110511$ & $0.9559$--$1.0206$ & blow-up & collapse & $2.110511$ & $0.9996$--$1.0001$ & $2.126345$ & $0.9826$--$1.0000$\\
$10^{-3}$ & $2.110511$ & $0.9947$--$1.0016$ & blow-up & collapse & $2.110511$ & $0.9999$--$1.0000$ & $2.114389$ & $0.9956$--$1.0000$\\
$5\cdot10^{-4}$ & $2.110511$ & $0.9993$--$1.0002$ & $2.255943$ & $0.8695$--$1.0000$ & $2.110511$ & $1.0000$--$1.0000$ & $2.111480$ & $0.9989$--$1.0000$\\
\bottomrule
\end{tabular*}
\caption{Rigid walls, $Ra=5000$, $k=3.117$, $t\in[0,4]$, Taylor--Hood elements on a $32\times32$ mesh graded towards the walls:
$Nu$ at $t=4$ (mean over the two walls) and the range of the computational scalar $\tilde Q^n$ over the run, for the weights $C=1$ and $C=100$ of the
$Q$ equation, with and without relaxation. The converged value on this mesh is $Nu=2.110511$. (The volume average of Section~\ref{sec:rolls} is not used here: without relaxation, the heat flux carries the factor $\tilde Q$.) With relaxation, the stored
$Q^n$ is $1$ at every step; without it, $Q^n=\tilde Q^n$. ``Blow-up'': $\frac{1}{2}\|\vu\|^2>10^{10}$ (the run is stopped);
``collapse'': $\min_n\tilde Q^n<10^{-6}$. $^\dagger$$Nu$ oscillates by more than $1\%$ over $t\in[3,4]$.}
\label{tab:dt}
\end{table}

\begin{figure}[htbp]
\centering
\includegraphics[width=\textwidth]{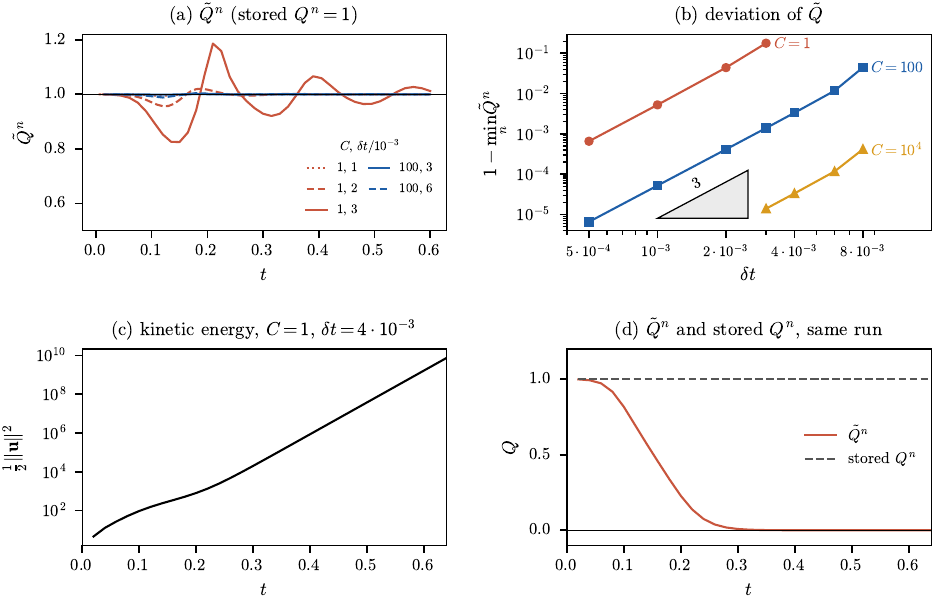}
\caption{The weight $C$ of the $Q$ equation (rigid walls, $Ra=5000$, $32\times32$ mesh, relaxed scheme). (a) The unrelaxed $\tilde Q^n$, with which the velocity and temperature fields are computed, during the formation of the rolls; the stored $Q^n$ is $1$ in all these runs. (b) The largest deviation $1-\min_n\tilde Q^n$ against $\dt$ for $C=1$, $100$, and $10^4$; the triangle shows slope $3$. (c), (d) One run with $C=1$, $\dt=4\cdot10^{-3}$: the kinetic energy grows without bound (c), while $\tilde Q^n$ collapses to $0$ and the stored $Q^n$ stays at $1$ (d).}
\label{fig:C}
\end{figure}

\item \emph{Accuracy during the transient period.} Table~\ref{tab:trans} and Figure~\ref{fig:conv}(b) compare the kinetic energy at $t=0.2$, when the rolls saturate, with a reference computation with a much smaller time-step.
With $C=100$ and relaxation, the error decreases at second-order and is two to four times smaller than with $C=1$---the difference is the error caused by $\tilde Q^n\ne1$.
Without relaxation, the drift of $Q$ dominates.
We see that with $C=1$, the error is already $35\%$ at $\dt=5\cdot10^{-4}$.

\begin{table}[htbp]
\centering
\footnotesize
\begin{tabular}{r cc cc}
\toprule
& \multicolumn{2}{c}{$C=1$} & \multicolumn{2}{c}{$C=100$}\\
\cmidrule(lr){2-3}\cmidrule(lr){4-5}
$\dt$ & relaxed & unrelaxed & relaxed & unrelaxed\\
\midrule
$4\cdot10^{-3}$ & $>1$ & $>1$ & $1.06{\times}10^{-2}$ & $1.72{\times}10^{-1}$\\
$2\cdot10^{-3}$ & $9.03{\times}10^{-3}$ & $>1$ & $2.42{\times}10^{-3}$ $(4.4)$ & $3.72{\times}10^{-2}$ $(4.6)$\\
$10^{-3}$ & $1.61{\times}10^{-3}$ $(5.6)$ & $>1$ & $5.89{\times}10^{-4}$ $(4.1)$ & $9.07{\times}10^{-3}$ $(4.1)$\\
$5\cdot10^{-4}$ & $2.73{\times}10^{-4}$ $(5.9)$ & $3.48{\times}10^{-1}$ & $1.45{\times}10^{-4}$ $(4.1)$ & $2.27{\times}10^{-3}$ $(4.0)$\\
\bottomrule
\end{tabular}
\caption{Formation of the rolls in the run of Table~\ref{tab:dt}: relative error of $\frac{1}{2}\|\vu\|^2$ at $t=0.2$, when the rolls saturate,
against a computation with $\dt=3.125\cdot10^{-5}$ and $C=100$ on the same mesh. In parentheses: ratio of the error at twice
the time step to this error ($4$ for second order).}
\label{tab:trans}
\end{table}

\item \emph{Faster flows.} The free-slip cells of Section~\ref{sec:rolls} at $r=1.5$ and $2$ show what goes wrong when $C$ is left at $1$ and the time-step is chosen for a slower flow.
With $\dt=0.02$, the run at $r=2$ blew up after $\tilde Q^n$ collapsed, and the run at $r=1.5$ settled into a spurious periodic cycle.
In both runs, the stored $Q^n$ was $1$ throughout, and only the range of $\tilde Q^n$ revealed the problem.
Therefore, we recommend choosing $C$ comparable to the kinetic energy of the flow, and reporting the range of $\tilde Q^n$, not the stored $Q^n$, as the diagnostic of a computation with the relaxed scheme.
\end{itemize}

\paragraph{Scaling with $\frac{E}{C}$}
The tests above vary $\dt$ and $C$ at a fixed Rayleigh number, so the kinetic energy $E$ of the rolls is fixed.
To test the dependence on $E$, we follow the formation of the rolls on the same mesh for $Ra\in\{2000,3000,5000,10^4,2\cdot10^4\}$, with $C\in\{1,100\}$ and $\dt$ from $2.5\cdot10^{-4}$ to $2\cdot10^{-3}$.
In each run, we record $\max_n|\tilde Q^n-1|$ until the rolls have settled (up to $t=0.5$, or $t=1.5$ at $Ra=2000$), and we observe the largest kinetic energy $E$ of these runs ranges from $11$ to $1600$.
We exclude seven runs in which $\tilde Q^n$ deviated from $1$ by more than $10\%$ or the solution blew up: with $C=1$, at $Ra=10^4$ for $\dt\ge10^{-3}$ and at $Ra=2\cdot10^4$ for $\dt\ge5\cdot10^{-4}$, and with $C=100$, at $Ra=2\cdot10^4$ for $\dt\ge10^{-3}$.

A least-squares fit of $\log\max_n|\tilde Q^n-1|$ against $\log\dt$, $\log C$, and $\log E$ over the remaining $33$ runs gives:
\begin{equation*}
\max_n|\tilde Q^n-1|\propto\dt^{3.04}\,C^{-1.00}\,E^{2.74}.
\end{equation*}
This fitted law reproduces the individual runs to within a typical factor of $1.25$ (root-mean-square of the logarithmic residuals).
With the advective time-scale $\tau=H/U$, where $U=\sqrt{2E/|\Omega|}$ (so that $E=\frac{1}{2}|\Omega|U^2$), the estimate of Section~\ref{sec:scheme} predicts
\begin{equation*}
\frac{E}{C}\Big(\frac{\dt}{\tau}\Big)^3\propto\dt^3\,C^{-1}\,E^{5/2}.
\end{equation*}
The fitted exponents of $\dt$ and $C$, $3.04$ and $-1.00$, are close to the predicted $3$ and $-1$, while the exponent of $E$, $2.74$, is larger than the predicted $\frac{5}{2}$. 

Figure~\ref{fig:scaling}(b) shows $\frac{C}{E}\max_n|\tilde Q^n-1|$ against $\big(\frac{\dt}{\tau}\big)^3$.
For $Ra\ge5000$ ($E\ge160$), all runs lie within $12\%$ of one line:
\begin{equation*}
\frac{C}{E}\max_n|\tilde Q^n-1|\approx15\Big(\frac{\dt}{\tau}\Big)^3 .
\end{equation*}
Closer to onset, at $Ra=3000$ and $2000$, the runs lie below this line by factors of $1.3$ and $3$.
There, the rolls form on the slower time-scale of the linear growth, so $U$ overestimates the speed of the flow and $H/U$ underestimates its time-scale.
This is consistent with the fitted exponent of $E$ being larger than $\frac{5}{2}$.
So, the deviation of $\tilde Q$ depends on $\dt$ and $C$ as $\frac{E}{C}\big(\frac{\dt}{\tau}\big)^3$ predicts, and on $E$ as predicted wherever $H/U$ is the right time-scale of the flow.
The estimate of Section~\ref{sec:scheme} is an order-of-magnitude argument that neglects the thermal part of $R^{n+1}$. In one rune per Rayleigh number ($C=100$, $A=1$), the largest thermal part is between $0.003\%$ and $0.02\%$ of the largest kinetic part, with the largest share near the onset.
The criterion $C\gtrsim E$ is a practical rule drawn from this scaling, not a condition of the analysis.
\begin{figure}[htbp]
\centering
\includegraphics[width=\textwidth]{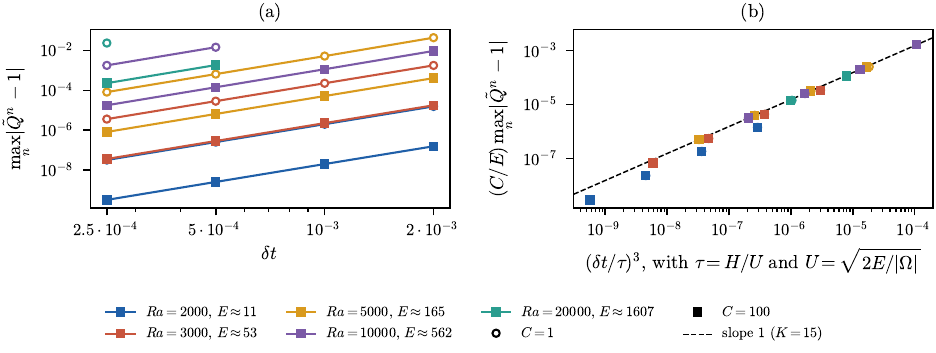}
\caption{Scaling of the deviation of the computational scalar (rigid walls, $32\times32$ mesh, relaxed scheme, formation of the rolls).
(a) $\max_n|\tilde Q^n-1|$ against $\dt$ for five Rayleigh numbers and $C=1$ (open) and $100$ (filled).
(b) The same, scaled by $\frac{C}{E}$, against $\big(\frac{\dt}{\tau}\big)^3$ with $\tau=H/U$ and $U=\sqrt{2E/|\Omega|}$; the dashed line has slope $1$.
Runs in which $\tilde Q^n$ deviated by more than $10\%$ or the solution blew up are omitted.}
\label{fig:scaling}
\end{figure}

\subsection{Energy stability versus practical time-step restrictions}\label{sec:practical}
Theorem~\ref{thm:stab} holds for every $\dt$ for the time discretization used in all runs of Section~\ref{sec:dt}.
Yet, the runs at $Ra=5000$ show three distinct restrictions on the time-step, and the theorem accounts for none of them.
\begin{enumerate}
\item \emph{Drift of the stored $Q$} (no relaxation).
With $C=1$, the steady-state is off by $7\%$ at $\dt=5\cdot10^{-4}$, and from $\dt=10^{-3}$ the solution grows without bound.
(The relaxation removes this drift at the cost of a few integrals per step.)
\item \emph{Deviation of $\tilde Q$ within a step.} $\tilde Q-1$ is $O\left(\frac{E}{C}\left(\frac{\dt}{\tau}\right)^3\right)$, and the relaxation does not change it.
With $C=1$, $\tilde Q$ collapses at $\dt=4\cdot10^{-3}$ on every mesh.
With $C=100$, $\tilde Q$ stays below $1.2\%$ up to $\dt=6\cdot10^{-3}$.
\item \emph{The explicit convection.} At $\dt\max|\nabla\vu|\approx1.2$--$1.6$ ($\dt\approx8\cdot10^{-3}$ here), the steady rolls are replaced by a growing oscillation, and the runs blow up at $1.2\cdot10^{-2}$ for $C=100$ and $C=10^4$ alike.
This limit is set by the flow---neither the relaxation nor $C$ removes it.
\end{enumerate}
However, our theorem is not contradicted.
Under heating, part (b) bounds the energy only by $e^{Kt}$ with $K\approx4G^2C_P^2/(\nu A)\approx10^7$ here, which allows growth on any time-scale of interest.
When $\tilde Q$ collapses, the convection and advection---multiplied by $\tilde Q$---are `switched off,' and the scheme integrates a linear problem whose convective instability nothing saturates.
Part (b) permits this kind of growth.
The energy of the scalar, $\frac{C}{4}\big(|Q^n|^2+|2Q^n-Q^{n-1}|^2\big)$, is about $\frac{C}{2}$ while $Q\approx1$ and drops to about $0$ when $Q$ collapses, so it never exceeds $O(C)$.
It cannot supply the growth, which comes from the buoyancy through the factor $e^{Kt}$.

We note: \emph{an unconditional energy bound is therefore not unconditional accuracy.} In practice, we use the relaxation, take $C\gtrsim E$, keep $\dt\max|\nabla\vu|$ below about $1$, and monitor the range of $\tilde Q^n$.
That range reveals the first two restrictions above, the drift of the stored scalar and the deviation of the computational scalar.
The third, the explicit convection, shows in the solution itself, while $\tilde Q^n$ stays close to $1$.

\subsection{Energy decay without heating}\label{sec:decay}
To test Theorem~\ref{thm:stab}(a), we take $T_b=T_t=T_{\rm ref}=0.5$, $G=3000$, rigid walls, and an initial perturbation of $\theta=0.2\sin(\pi y)\cos(kx)$, which drives a transient flow that then decays.
Here, $10B'/\kappa\approx4.9\times10^5$ for $\dt=10^{-2}$.
Table~\ref{tab:decay} shows that $X^n$ decreases at every step for $\dt$ from $10^{-2}$ to $10^{-4}$, both with $A=5\times10^5$ (as part (a) of Theorem~\ref{thm:stab} states for the time discretization) and with $A=1$.
The largest step is stable, but not accurate---we see its peak kinetic energy is $17\%$ too large.

\begin{table}[htbp]
\centering
\footnotesize
\begin{tabular}{r cc cc}
\toprule
& \multicolumn{2}{c}{$A=5\times10^5$ (theorem applies)} & \multicolumn{2}{c}{$A=1$}\\
\cmidrule(lr){2-3}\cmidrule(lr){4-5}
$\dt$ & $\max_n(X^{n+1}-X^n)$ & $\max_n\frac{1}{2}\|\vu^n\|^2$ & $\max_n(X^{n+1}-X^n)$ & $\max_n\frac{1}{2}\|\vu^n\|^2$\\
\midrule
$10^{-2}$ & $-6.25{\times}10^{-6}$ & $7.542$ & $-2.17{\times}10^{-6}$ & $7.540$\\
$10^{-3}$ & $-4.67{\times}10^{-7}$ & $6.468$ & $-2.97{\times}10^{-8}$ & $6.468$\\
$10^{-4}$ & $-4.51{\times}10^{-8}$ & $6.461$ & $-1.20{\times}10^{-9}$ & $6.461$\\
\bottomrule
\end{tabular}
\caption{No differential heating ($T_b=T_t=T_{\rm ref}$), $G=3000$, rigid walls, relaxed scheme, Taylor--Hood elements on a
$32\times32$ mesh: largest one-step change of the energy $X^n$ of Theorem~\ref{thm:stab}(a) over $t\in[0,0.5]$ (steps $n\ge3$),
and the peak kinetic energy.}
\label{tab:decay}
\end{table}

\subsection{Sensitivity to the relaxation parameter}\label{sec:eta}
All computations above use $\eta=0.95$.
Table~\ref{tab:eta} repeats the time-step study of Table~\ref{tab:dt} with $C=1$ and the decay test of Table~\ref{tab:decay} for $\eta$ between $0$ and $0.99$.

For $\eta\ge0.5$, the relaxation resets the stored $Q^n$ to $1$ at every step and every time-step.
The results agree to all digits shown, and the same holds for $\eta=10^{-2}$ up to $\dt=10^{-3}$.
By Lemma~\ref{lem:contract}, $\eta$ matters only when the budget $\eta\,\dt\,\mathcal D^{n+1}$ is comparable to the change of $\mathcal E_Q$ needed to reset $Q$, about $\frac{3C}{2}|\tilde Q^{n+1}-1|$, and this change is tiny compared with the dissipation of a step.

With $\eta=0$, the budget allows no relaxation in these runs, and $Q$ drifts downwards.
Moving $\tilde Q^{n+1}<1$ towards $1$ increases $\mathcal E_Q$.
So, with $\eta=0$, the scheme is identical to the unrelaxed scheme.
With $\eta=10^{-8}$ and $10^{-6}$, the budget is so small that the results barely differ from those with $\eta=0$.
For $\eta=10^{-4}$, the relaxation is partial, but not enough; $Q^n$ falls to $0.15$ at $\dt=10^{-3}$, to $0.14$ at $\dt=3\cdot10^{-3}$, and then the solution blows up.
At $\eta=10^{-2}$ and $\dt=3\cdot10^{-3}$ though, $Q^n$ dips to $0.939$ during the transient period and returns to $1$, and the steady-state is again exact.

In the decay test, $X^n$ decreases at every step for every $\eta$, as part (a) of Theorem~\ref{thm:stab} states for the time discretization.
The largest one-step change varies by only $10\%$.
For the manufactured solution of Section~\ref{sec:mms} ($\dt=0.0125$), the errors for $\eta\in\{0.5,0.9,0.95,0.99\}$ coincide to all twelve digits shown, and $\eta=0$ reproduces the unrelaxed values.
In these tests, any $\eta\in[0.5,1)$ therefore gives the same results.
The theorem requires only $\eta<1$, and a smaller $\eta$ keeps a larger part $(1-\eta)\mathcal D^{n+1}$ of the dissipation in the energy estimate.

All of this concerns the stored $Q^n$, which is controlled by $\eta$.
The unrelaxed $\tilde Q^n$, with which the velocity and temperature fields are computed, is the same for all $\eta\ge0.5$ (it depends on $\eta$ only through the stored values $Q^n$, $Q^{n-1}$, which are then $1$), and its accuracy is controlled by $\dt$ and $C$ (Section~\ref{sec:dt}).

\begin{table}[htbp]
\centering
\setlength{\tabcolsep}{3.5pt}\footnotesize
\begin{tabular}{c cc cc c}
\toprule
& \multicolumn{2}{c}{$\dt=3\cdot10^{-3}$} & \multicolumn{2}{c}{$\dt=10^{-3}$} & decay\\
\cmidrule(lr){2-3}\cmidrule(lr){4-5}\cmidrule(lr){6-6}
$\eta$ & $Nu(4)$ & $\min Q^n$ & $Nu(4)$ & $\min Q^n$ & $\max_n(X^{n+1}-X^n)$\\
\midrule
$0$ & blow-up & $<10^{-2}$ & blow-up & $<10^{-2}$ & $-2.39{\times}10^{-6}$\\
$10^{-8}$ & blow-up & $<10^{-2}$ & blow-up & $<10^{-2}$ & $-2.39{\times}10^{-6}$\\
$10^{-6}$ & blow-up & 0.0280 & blow-up & 0.0279 & $-2.39{\times}10^{-6}$\\
$10^{-4}$ & blow-up & 0.1350 & blow-up & 0.1463 & $-2.39{\times}10^{-6}$\\
$10^{-2}$ & 2.110511 & 0.9392 & 2.110511 & 1.0000 & $-2.25{\times}10^{-6}$\\
$0.5$ & 2.110511 & 1.0000 & 2.110511 & 1.0000 & $-2.17{\times}10^{-6}$\\
$0.95$ & 2.110511 & 1.0000 & 2.110511 & 1.0000 & $-2.17{\times}10^{-6}$\\
$0.99$ & 2.110511 & 1.0000 & 2.110511 & 1.0000 & $-2.17{\times}10^{-6}$\\
\bottomrule
\end{tabular}
\caption{Sensitivity to the relaxation parameter $\eta$ (Taylor--Hood elements, $32\times32$ mesh). Left: rigid walls, $Ra=5000$,
$k=3.117$, $t\in[0,4]$, $C=1$ as in Table~\ref{tab:dt}: $Nu$ at $t=4$ (converged value $2.110511$ on this mesh) and the
smallest stored $Q^n$. Right: no differential heating, $A=1$, $\dt=10^{-2}$ as in Table~\ref{tab:decay}: largest one-step
change of $X^n$.}
\label{tab:eta}
\end{table}

\subsection{A stress test at $Ra=10^6$}\label{sec:ra1e6}
This run is a stress test of the scheme in an unsteady flow with thin boundary layers and plumes.
It establishes that, with $C$ above the kinetic energy and a time-step below the explicit-convection limit, $\tilde Q^n$ stays accurate and the temperature bounded through the eruption of plumes, and it shows the failure mode when these conditions are violated.

With this test, we do not establish converged statistics, mesh convergence at this Rayleigh number, or long-time behavior.
We take $Ra=10^6$, $Pr=1$ ($\nu=\kappa=1$, $G=Ra$), rigid walls, and a periodic cell of aspect ratio $L=4$, on a $128\times48$ mesh graded towards the walls (wall-spacing $1.1\cdot10^{-3}$), and start from the conduction state with a random temperature perturbation of amplitude $10^{-2}$.
Times are in units of the diffusion time $H^2/\kappa$, in which the free-fall time $H/\sqrt{G\,|T_b-T_t|\,H}$ is $Ra^{-1/2}=10^{-3}$.
We use $\dt=2\cdot10^{-6}$ and $C=10^6$, larger than the kinetic energy at all times, and integrate to $t=0.03$.

Figure~\ref{fig:ra1e6} shows the evolution of the flow.
The perturbation grows in the unstable conduction state, and at $t\approx0.009$, six pairs of plumes erupt from the thermal boundary layers.
The plumes cross the layer and roll up.
Meanwhile, the wall Nusselt numbers leave $1$ at $t\approx0.0085$ and peak at $11.9$ at the bottom wall ($t=0.0117$) and $12.7$ at the top wall ($t=0.0124$).
The flow then reorganizes into three to four large cells, with thin boundary layers and plumes along their edges.
Over $0.015\le t\le0.03$, the Nusselt number at the bottom wall varies between $6.4$ and $8.6$, the mean at both walls is $7.6$, the kinetic energy is about $1.6\cdot10^5$, and the largest speed is about $600$.
The volume average of Section~\ref{sec:rolls} agrees with the wall values only in a steady-state.
While the plumes carry heat into the interior faster than the walls supply it, the volume average peaks at $43$, and over $0.015\le t\le0.03$, its mean is $7.5$.
Because the run is too short for converged statistics, we make no claim about the flow beyond $t=0.03$.
Throughout the run, $\tilde Q^n$ stays within $5\cdot10^{-9}$ of $1$, the stored $Q^n$ is $1$ after every relaxation, and the temperature in the snapshots stays in $[0,1]$ (between $0.004$ and $0.996$ in the interior).

\begin{figure}[htbp]
\centering
\includegraphics[width=\textwidth]{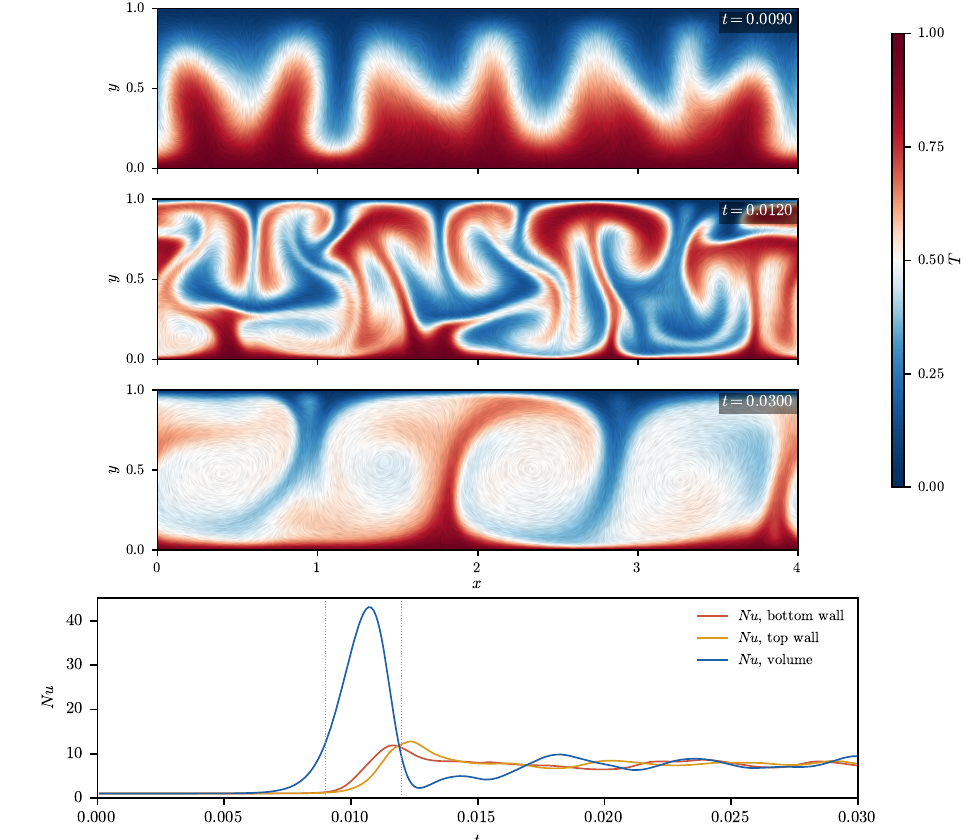}
\caption{Convection at $Ra=10^6$, $Pr=1$, rigid walls, $L=4$, $128\times48$ mesh, $\dt=2\cdot10^{-6}$, $C=10^6$.
Top three panels: temperature (contrast stretched around $T=1/2$) with a line-integral-convolution texture of the velocity, during the eruption of the plumes ($t=0.009$), after their roll-up ($t=0.012$), and in the large-scale circulation ($t=0.03$).
Bottom: the Nusselt number at the two walls and the volume average; the dotted lines mark the first two panels.
The post-processing of this figure is described in the appendix.}
\label{fig:ra1e6}
\end{figure}

A first attempt with $\dt=10^{-5}$ and $C=10^5$ shows the failure mode of Section~\ref{sec:dt}.
Up until $t=0.0102$, $1-\tilde Q^n$ was below $3\cdot10^{-6}$.
Within the next ten steps though, as the plumes erupted and $\max |\nabla\vu|$ grew roughly fourfold, $\tilde Q^n$ fell to $0.32$.
The relaxation step could not restore the stored $Q^n$---which followed it to $0.34$---and the computation blew up at $t=0.0246$.
The collapse was sudden rather than a gradual drift, as at the explicit-convection limit in Section~\ref{sec:dt}.
This run does not separate the two causes, however.
Its value for $C$ was also below the kinetic energy at the collapse ($1.3\cdot10^5$).
The successful run satisfies {\em both} conditions: its time-step, $2\cdot10^{-6}$, is five times smaller than that of the failed attempt, and $C=10^6$ is larger than the kinetic energy.
As in Section~\ref{sec:dt}, the range of $\tilde Q^n$ is what reveals the failure.

\subsection{Three dimensions}\label{sec:3d}
Remark~\ref{rem:3d} shows that the analysis holds unchanged in a three-dimensional channel with flat walls.
This section describes a three-dimensional implementation, the checks it passes, and two computations with it.

\paragraph{Implementation} The domain is $[0,L_x)\times(0,H)\times[0,L_z)$, periodic in $x$ and $z$, with $y$ vertical.
The walls carry the Navier-slip condition on both tangential components, which contributes $l\int_\G(u_xv_x+u_zv_z)\,ds$ to the bilinear forms, and $u_y=0$, $\theta=0$.
The mesh has $N_x\times N_y\times N_z$ hexahedra, each split into six tetrahedra, with the points in $y$ placed the same as in two dimensions for rigid walls.
Everything else is the same as in Section~\ref{sec:impl}: Taylor--Hood elements with $P_2$ temperature, the skew-symmetric form \eqref{skew}, and the hydrostatic pressure of the conduction state removed analytically.
A step solves eight scalar Helmholtz problems now (three components each of $\vu_1$, $\vu_2$, then $\theta_1$, $\theta_2$) with three matrices, since $u_x$ and $u_z$ share one, followed by one pressure Poisson problem, and three $L^2$ projections.
The relaxation uses the budget of Remark~\ref{rem:3d}.
The right-hand sides are products with matrices assembled once (one each for mass, gradient, and divergence) and with one convection matrix assembled per step.
The integrals $I_1$, $I_2$ of \eqref{Qsol} are inner products with these right-hand sides.

In three dimensions, the sparse direct factorization of the $P_2$ matrices is expensive with memory. So, we solve with conjugate gradients, with a relative tolerance $10^{-12}$.
On the manufactured solution and the rolls of the checks below, the errors and diagnostics computed with conjugate gradients agree with those of the direct solver to $10$--$12$ digits.

\paragraph{Checks} The implementation passes the following tests:
\begin{itemize}
\item \emph{Manufactured solution.} \ref{app:mms} ($\nu=0.1$, $\kappa=0.05$, $G=1$, $l=1$), the exact solution extends \eqref{mms} with a second stream function $\chi$ and a dependence on $z$:
\begin{align}\label{mms3dsol}
&\vu=\big(\partial_y\psi,\,-\partial_x\psi-\partial_z\chi,\,\partial_y\chi\big),\\[2pt]
&\psi=\sin t\sin x\,(1-y^2),\\[2pt]
& \chi=\tfrac{1}{2}\sin t\sin z\,(1-y^2),\\[2pt]
&p=\sin t\cos x\sin y\cos z,\\[2pt]
&T=\tfrac{1-y}{2}+\tfrac{1}{2}\cos t\,(1-y^2)\big(\sin x+\tfrac{y}{2}\cos x+\tfrac{1}{2}\cos z\big).
\end{align}
The velocity is divergence free, $u_y=0$ on the walls, and both tangential components slip at the walls.
The forcing of the Navier-slip condition is imposed on both.
Table~\ref{tab:mms3d} shows second-order in time for $u$, $T$, and $p$, and spatial orders of about $3$ for $u$ and $T$ and between $2$ and $2.5$ for $p$, as in two dimensions (Table~\ref{tab:space}).
\item \emph{Isotropy.} For free-slip walls at $Ra=1000$ and $k=\pi/\sqrt2$, a steady roll varying in $x$ ($12\times12\times2$ mesh, $L_z=0.5$) and the same roll varying in $z$ ($2\times12\times12$ mesh, $L_x=0.5$) grow at the same rate, $\sigma=3.446650$ ($\dt=0.005$).
The two-dimensional code on the same $12\times12$ cross-section gives $3.446617$, and the exact value is $3.4530$.
The difference from it is the error of this coarse cross-section, shared by both codes.
\item \emph{Rolls form independent of $z$.} Rigid rolls at $Ra=5000$, $Pr=1$, $k=3.117$ ($12\times12\times2$ mesh, $L_z=0.5$, $\dt=2\cdot10^{-3}$, $C=100$) reach the steady-state by $t=4$, with $Nu=2.1103080$ (volume average) against $2.1103008$ from the two-dimensional code on the $12\times12$ mesh.
The kinetic energy per unit length in $z$ is $142.0568$ against $142.0622$, and the out-of-plane velocity stays exactly zero.
\item \emph{Energy decay without heating.} With the setting of Section~\ref{sec:decay} ($T_b=T_t=T_{\rm ref}$, $G=3000$, rigid walls) on the $8^3$ mesh, with $L_z=L_x$ and the initial perturbation $\theta=\sin(\pi y)\cos(kx)\big(0.2+0.1\cos(kz)\big)$ (its $z$-independent part is the perturbation of Section~\ref{sec:decay}, and the second term makes the initial state three-dimensional), $X^n$ decreases at every step for $A=1$ and $5\cdot10^5$, and for $\dt=10^{-2}$ and $10^{-3}$, as Theorem~\ref{thm:stab}(a) states for the time discretization.
\end{itemize}

\begin{table}[htbp]
\centering\small\setlength{\tabcolsep}{3.7pt}
\begin{tabular}{rccccccc}
\toprule
\multicolumn{8}{l}{(a) Time, $6^3$ mesh: $L^2$ differences from the run with $\dt/32$}\\
$\dt$ & $u$ & order & $T$ & order & $p$ & order & \\
\midrule
$0.05$  & $7.14\cdot10^{-3}$ & 1.98 & $9.43\cdot10^{-3}$ & 2.02 & $6.30\cdot10^{-3}$ & 1.98 & \\
$0.025$  & $1.80\cdot10^{-3}$ & 2.01 & $2.33\cdot10^{-3}$ & 2.02 & $1.60\cdot10^{-3}$ & 2.00 & \\
$0.0125$ & $4.49\cdot10^{-4}$ & 2.06 & $5.72\cdot10^{-4}$ & 2.07 & $3.99\cdot10^{-4}$ & 2.06 & \\
$0.00625$ & $1.07\cdot10^{-4}$ &   & $1.36\cdot10^{-4}$ &   & $9.54\cdot10^{-5}$ &   & \\
\midrule
\multicolumn{8}{l}{(b) Space, $N^3$ mesh, $\dt=0.0125$: $L^2$ errors against the exact solution at $t=0.5$}\\
$N$ & $u$ & order & $T$ & order & $p$ & order & \\
\midrule
$4$ & $2.33\cdot10^{-1}$ &   & $7.97\cdot10^{-2}$ &   & $4.17\cdot10^{-1}$ &   & \\
$6$ & $9.06\cdot10^{-2}$ & 2.33 & $2.49\cdot10^{-2}$ & 2.86 & $1.56\cdot10^{-1}$ & 2.43 & \\
$8$ & $3.74\cdot10^{-2}$ & 3.07 & $1.08\cdot10^{-2}$ & 2.89 & $7.68\cdot10^{-2}$ & 2.46 & \\
$12$ & $9.23\cdot10^{-3}$ & 3.45 & $3.29\cdot10^{-3}$ & 2.94 & $3.13\cdot10^{-2}$ & 2.22 & \\
\bottomrule
\end{tabular}
\caption{Three-dimensional manufactured solution (Section~\ref{sec:3d}), $t\in[0,0.5]$.
The meshes of (b) are still coarse, and the orders are pre-asymptotic.}
\label{tab:mms3d}
\end{table}

\paragraph{Three-dimensional perturbations of steady rolls} We start from the rigid rolls of Table~\ref{tab:cb} ($Pr=0.71$, $Ra=5000$, $\alpha=3.117$) in a square periodic box $L_x=L_z=\frac{2\pi}{\alpha}$, $H=1$, and add a small three-dimensional mode:
\begin{equation*}
\theta=\sin(\pi y)\cos(\alpha x)\big(0.1+0.01\cos(\alpha z)\big).
\end{equation*} We use a $16^3$ mesh, $\dt=2\cdot10^{-3}$ and $C=10^3$, and integrate to $t=1$, by which time the flow is steady. The box only allows perturbations whose wave number in $z$ is a multiple of $\alpha$, so this is a test of the scheme on a flow whose two-dimensional state is known, not a stability analysis of the rolls. We monitor the kinetic energy of $u_z$, which measures the three-dimensional part of the flow, the Nusselt number, and $\tilde Q^n$ (Figure~\ref{fig:caseA}).

The three-dimensional mode first grows along with the rolls, and its kinetic energy peaks near $t=0.12$.
Once the rolls saturate, the energy decays exponentially (Figure~\ref{fig:caseA}(b)), and from $t\approx0.7$ it stays at the level of
round-off and solver tolerance, at most $3\cdot10^{-13}$ of the kinetic energy of the rolls.
In this box, the two-dimensional rolls are therefore stable to the imposed perturbation, and the three-dimensional computation returns to them.
From $t=0.8$ on, the Nusselt number and the kinetic energy no longer change.
The final state agrees with the two-dimensional code on the same $16\times16$ cross-section,
\begin{equation*}
\begin{aligned}
Nu&=2.1101424\ \text{(3D)}, &\qquad &2.1101416\ \text{(2D)},\\
\tfrac{1}{2}\|\vu\|^2&=283.5888\ \text{(3D)}, &&283.5911\left( \text{2D, for $L_z=\tfrac{2\pi}{\alpha}$}\right),
\end{aligned}
\end{equation*}
to better than $10^{-5}$ relative difference, even though the two meshes do not overlap (the three-dimensional mesh consists of tetrahedra and does not contain the two-dimensional one).
The wall values of $Nu$ differ from the volume average by $3\cdot10^{-4}$ relative difference, the error of the wall gradient discussed in Section~\ref{sec:rolls}.
Throughout the run, $\tilde Q^n$ stays within $5\cdot10^{-5}$ of $1$, and once the flow is steady, from $t=0.9$ on, $|\tilde Q^n-1|<10^{-12}$. %The run took $3.3$\,s per step on a laptop.

\begin{figure}[htbp]
\centering
\includegraphics[width=\textwidth]{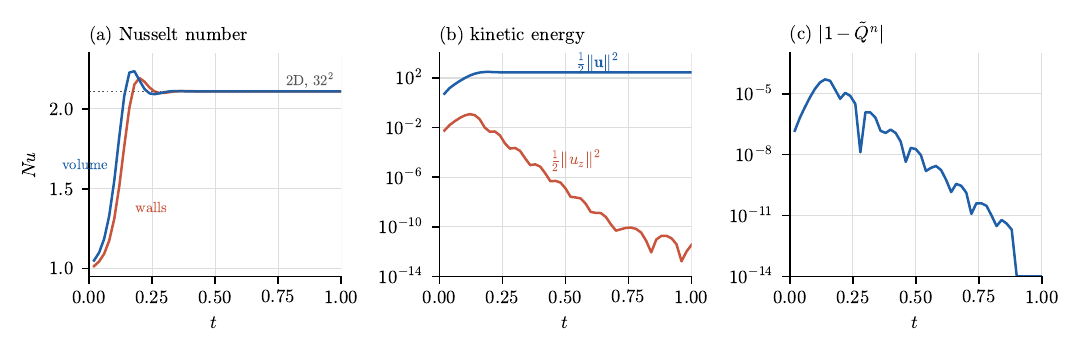}
\caption{Three-dimensional perturbation of steady rolls: $Pr=0.71$, $Ra=5000$, rigid walls, $L_x=L_z=2\pi/\alpha$, $\alpha=3.117$, $16^3$ mesh, $\dt=2\cdot10^{-3}$, $C=10^3$, initial temperature $\sin(\pi y)\cos(\alpha x)(0.1+0.01\cos(\alpha z))$. (a) Nusselt number at the walls and volume average; the dotted line is the two-dimensional value on the $32\times32$ mesh (Table~\ref{tab:cb}). (b) Kinetic energy and the kinetic energy of $u_z$. (c) Deviation of $\tilde Q^n$ from $1$ (below $10^{-12}$ after $t=0.9$).}
\label{fig:caseA}
\end{figure}

\paragraph{Convection from random initial data} In a wide box, $L_x=L_z=4$, $H=1$ ($32\times16\times32$ mesh), with $Pr=0.71$, $Ra=5000$, and rigid walls, we start from the no-flow conductive steady-state with a random temperature perturbation of amplitude $10^{-3}$ (uniform values at the nodes, multiplied by $\sin(\pi y)$, seed $1$), and integrate to $t=6.76$ ($\dt=2\cdot10^{-3}$, $3380$ steps), by which time the flow is steady.
Here, $C=10^4$, below the free-fall bound $\frac{1}{2}|\Omega|U_{\rm ff}^2=2.8\cdot10^4$ of Section~\ref{sec:scheme}.
We describe the patterns by how many wavelengths fit across the box.
A mode $(m,n)$ has $m$ wavelengths along $x$ and $n$ along $z$; its wave vector is $\frac{\pi}{2}(m,n)$, and its wavelength is $4/\sqrt{m^2+n^2}$.
For example, the mode $(2,0)$ consists of four straight rolls with axes along $z$, and the modes $(1,1)$ and $(1,-1)$ consist of rolls along the two diagonals of the box.
We find the dominant modes from the vertical velocity on the mid-plane $y=\frac{1}{2}$.
Figure~\ref{fig:caseB} shows three stages and the time evolution.
\begin{itemize}
\item \emph{Growth and saturation.} The perturbation grows exponentially, and the kinetic energy increases rapidly before saturating at $t\approx0.65$. $Nu$ reaches $2.09$ at $t=0.66$, dips to $2.03$, and then settles near $2.10$.
\item \emph{Rolls of wave number $\pi$.} From $t\approx1$ to $2$, the flow consists of four straight rolls with axes along $z$, the mode $(2,0)$.
This mode carries $87$--$93\%$ of the mid-plane spectrum, and $Nu=2.08$--$2.10$.
The rolls are not exactly straight; the kinetic energy of $u_z$, $\frac{1}{2}\|u_z\|^2$, is $1$--$2.5\%$ of the total.
\item \emph{Transition.} Two oblique modes, $(3,1)$ and $(1,-1)$, then grow on the rolls, from $4\%$ of the mid-plane spectrum at $t=1.5$ to $17\%$ at $t=2.5$.
Around $t=3$, the rolls break up into a cellular pattern (Figure~\ref{fig:caseB}(b)), and $Nu$ drops to its minimum, $1.85$.
\item \emph{Diagonal rolls.} From $t\approx4.5$, a single mode, $(1,-1)$, carries more than $99.5\%$ of the mid-plane spectrum.
The flow now consists of straight rolls along the diagonal of the box, with wave number $\frac{\pi}{\sqrt2}\approx2.22$.
It is steady from $t\approx5$: over the rest of the run, $Nu$ and the kinetic energy change by about one part in $10^6$.
\end{itemize}
Because these rolls are two-dimensional in a rotated frame, the two-dimensional code provides a check on the final state.
We compute rolls of wave number $\frac{\pi}{\sqrt2}$ (period $L=2\sqrt2$) on a $24\times16$ mesh, whose spacing along the wave vector is close to that of the three-dimensional mesh.
The two codes agree:
\begin{equation*}
\begin{aligned}
Nu&=1.9504043\ \text{(3D)}, &\qquad &1.9504015\ \text{(2D)},\\
\tfrac{1}{2}\|\vu\|^2\ \text{per unit area}&=60.101\ \text{(3D)}, &&60.096\ \text{(2D)}.
\end{aligned}
\end{equation*}
On a finer $48\times32$ mesh, the two-dimensional code gives $Nu=1.950454$.
So, the difference between the codes on comparable meshes, $3\cdot10^{-6}$ in $Nu$, is well below the error of the spatial discretization, about $5\cdot10^{-5}$.
Throughout the run, $\tilde Q^n$ stays within $1.3\cdot10^{-5}$ of $1$, with the largest deviation at saturation, and within $10^{-10}$ once the flow is steady.
The stored $Q^n$ is $1$ at every step.

We do not claim that this sequence of patterns is generic.
The box admits only modes with whole numbers of wavelengths across it, and we have not examined the stability of the rolls of wave number $\pi$ to other perturbations. 

\begin{figure}[htbp]
\centering
\includegraphics[width=\textwidth]{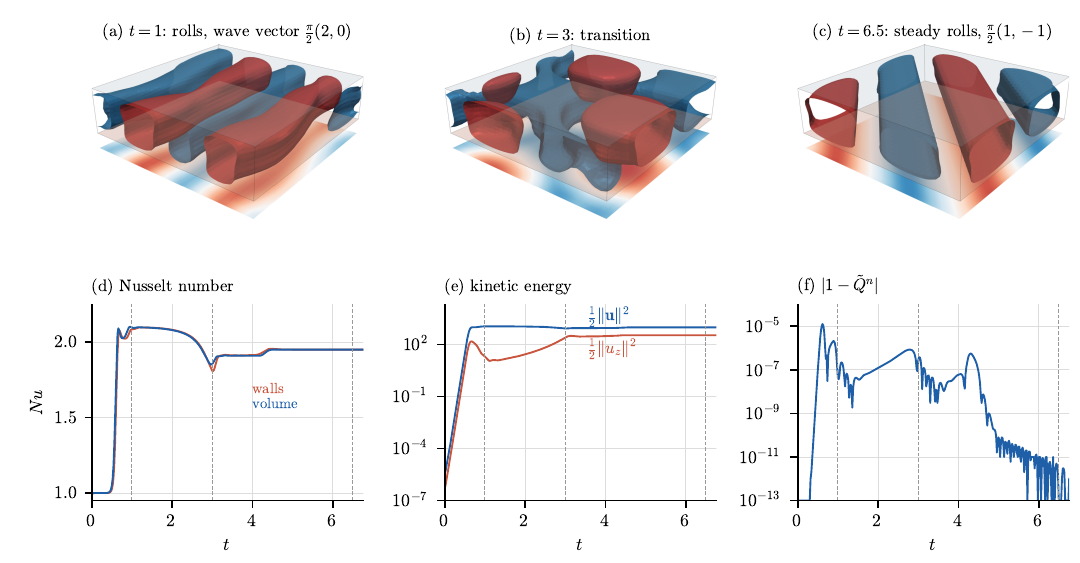}
\caption{Convection from random initial data in the box $4\times1\times4$: $Pr=0.71$, $Ra=5000$, rigid walls, $32\times16\times32$ mesh, $\dt=2\cdot10^{-3}$, $C=10^4$. (a)--(c) Isosurfaces of the temperature anomaly $T-\langle T\rangle(y)$ at $\pm0.12$ (where $\langle T\rangle$: horizontal mean with red as warm, rising fluid and blue as cold, sinking fluid). (a) straight rolls, mode $(2,0)$, (b) the transition, (c) the steady rolls along the diagonal, mode $(1,-1)$. (d) Nusselt number at the bottom wall and volume average. (e) Kinetic energy and the kinetic energy of $u_z$. (f) Deviation of $\tilde Q^n$ from $1$.
Dashed lines: the times of (a)--(c).
The snapshots were written as element averages; for the images these are averaged to the mesh vertices, interpolated linearly to a uniform grid of spacing $1/96$, and the isosurfaces lightly smoothed.}
\label{fig:caseB}
\end{figure}

\section{Concluding remarks}
We presented a linear, second-order and fully decoupled scheme, with four constant-coefficient Helmholtz problems and one pressure Poisson problem per step.
For its time discretization, with exact spatial operators, we proved that for any time-step, the energy decays without heating, under a condition on the thermal weight, and stays bounded on finite time intervals under Rayleigh--B\'enard heating.
While this is not proved for the Taylor--Hood finite element implementation, all our tests without heating showed the discrete energy decreased at every step.
However, as the experiments show, energy stability for every time-step is not accuracy for every time-step.

The scheme differs from related SAV schemes for the Boussinesq equations in three ways.
The GSAV schemes of \cite{AlhomsiWuZheng2026, WagnerWohlmuthZawallich2026} use an energy-based scalar with an exponential integrator and present error estimates in settings without heat input. The stabilized ensemble schemes of \cite{JiangYang2023anm} use two exponential scalars and are proved long-time stable for homogeneous boundary conditions.
Our scheme instead uses a single zero-energy scalar, whose exact value is known and can be restored by relaxation.
Our scheme considers Rayleigh--B\'enard configuration and Navier-slip walls, as verified against: classical onset, weakly nonlinear results, and published Nusselt numbers of steady rolls.
Section~\ref{sec:noerr} explains what stands in the way of an error estimate: growth constants intrinsic to heating from below, the explicit convection, the rotational pressure correction with slip walls, and the relaxation.

The experiments separate three effects that energy stability does not control.
First, the stored zero-energy variable drifts and the relaxation removes this drift at negligible cost.
Without relaxation, the drift biases the computed steady-states and, in convection, can make the solution blow up at time-steps far below the limit set by the explicit terms.
Second, the computational scalar $\tilde Q$ deviates from $1$, and since the velocity and temperature are computed with it, the relaxation cannot correct this deviation.
For a flow with kinetic energy $E$ and time-scale $\tau$, the deviation is of order $\frac{E}{C}\big(\frac{\dt}{\tau}\big)^3$.
In convection, $E$ grows with the Rayleigh number, so the deviation is amplified unless $C$ is comparable to $E$.
Choosing $C$ this way adds no linear solves and leaves the stability analysis unchanged.
In our tests at $Ra=5000$, it reduced the error while the rolls formed by a factor of two to four, and it doubled the largest time-step at which the steady rolls were reproduced (Section~\ref{sec:dt}).
Additionally, at $Ra=10^6$, it kept $\tilde Q$ within $5\cdot10^{-9}$ of $1$ through the eruption of plumes.
With $C=1$, by contrast, a steady-state turned into a spurious periodic cycle that only the range of $\tilde Q$ revealed.
Third, the explicit convection limits the time-step to $\dt\max|\nabla\vu|=O(1)$, a limit set by the flow rather than by the grid.
Therefore, we recommend using the range of $\tilde Q$, not the stored $Q$, as a diagnostic for accuracy.
Stabilization terms, such as those of \cite{JiangYang2023anm,JiangYang2024jmaa}, are added to SAV schemes to improve
their accuracy for complex flows.
Part of the inaccuracy these stabilization terms address may be the scale effect identified here: when the energy of the scalar is of order one but the kinetic energy of the flow is large, the deviation of the scalar is amplified by their ratio, and choosing the weight $C$ comparable to the kinetic energy removes this amplification.
We have not examined whether this is so for those schemes, and this remains a direction for possible research.

In the tests noted here, the results are not sensitive to the relaxation parameter $\eta \in [0.5,1)\,.$
The stability analysis holds unchanged in a three-dimensional channel with flat walls (Remark~\ref{rem:3d}), and a three-dimensional implementation passes corresponding verification tests (Section~\ref{sec:3d}).
In three dimensions, a perturbation of steady rolls decays to round-off error, and convection from random data in a wide box passes through a change of roll pattern to steady diagonal rolls that agree with the two-dimensional code.
In both computations, $\tilde Q$ stays within $10^{-4}$ of $1$.

Natural next steps for this work could be: a study of the three-dimensional instabilities of the rolls \cite{CleverBusse1974} in larger boxes; walls that are curved, rough, or heated from the side, and obstructions inside the cell---which can replace steady convection by periodic flow, as observed in \cite{MhinaSouleyHassane2022}; and a fluid layer coupled to a porous medium, in which convection may stay in the fluid or penetrate the porous layer \cite{McCurdyMooreWang2019,McCurdyMooreWang2022}.
For the coupled case described with those works, we speculate the zero-energy variable would have to absorb the coupling terms on the fluid--porous interface as it absorbs the convection terms here, and the rotational pressure correction for the Stokes--Darcy system \cite{JiangLi2024} is a natural starting point.

\appendix
\section{The manufactured solution}\label{app:mms}
\paragraph{Problem} $\Omega=[0,2\pi)\times(-1,1)$, periodic in $x$, walls at $y=\pm 1$ ($H=2$), $\nu=0.1$, $\kappa=0.05$,
$G=1$, $T_{\rm ref}=0.5$, $l=1$, $T_b=1$ at $y=-1$ and $T_t=0$ at $y=1$.
The exact solution is
\begin{gather*}
\vu=\big(\partial_y\psi,-\partial_x\psi\big)=\sin t\,\big(-2y\sin x,\,-(1-y^2)\cos x\big),\\
\psi=\sin t\sin x\,(1-y^2),\qquad p=\sin t\cos x\sin y,\\
T=\tfrac{1-y}{2}+\tfrac{1}{2}\cos t\,(1-y^2)\big(\sin x+\tfrac{y}{2}\cos x\big).
\end{gather*}
It satisfies $\nabla\cdot\vu=0$ and, on the walls, $u_y=0$, $T=T_b$ at $y=-1$, and $T=T_t$ at $y=1$.
The forcing added to the momentum and heat equations, $\vu_t+(\vu\cdot\nabla)\vu-\nu\Delta\vu+\nabla p-G(T-T_{\rm ref})\ey=(f_x,f_y)$ and $T_t+\vu\cdot\nabla T-\kappa\Delta T=f_T$, is
\begin{align*}
f_x&=\sin x\,\big(2(1+y^2)\sin^2t\cos x-2y\cos t-\sin t\sin y-2\nu y\sin t\big),\\
f_y&=-(1-y^2)\cos t\cos x-2y(1-y^2)\sin^2t+\sin t\cos x\cos y\\
&\quad+\nu(y^2-3)\sin t\cos x-G(T-T_{\rm ref}),\\
f_T&=\tfrac{1-y^2}{4}\sin t\,\big(2y^2\cos t\sin^2x+(3y^2-1)\cos t\cos^2x\\
&\qquad-y\cos x-2\sin x+2\cos x\big)\\
&\quad+\tfrac{\kappa}{4}\cos t\,\big(6\sin x-2y^2\sin x+7y\cos x-y^3\cos x\big),
\end{align*} and the Navier-slip condition becomes: $l u_x+\nu\partial_\vn u_x=g_w$ with
\begin{align*}
g_w&=2(l+\nu)\sin t\sin x &&\textrm{at}\quad y=-1,\\
g_w&=-2(l+\nu)\sin t\sin x &&\textrm{at}\quad y=1.
\end{align*} (These expressions were generated and checked symbolically. 
The script that writes them into the FreeFEM code is in the supplementary material.) In the scheme, the forcing is evaluated at $t_{n+1}$ and added to \eqref{su} and \eqref{sT} and to the slip condition. 
For the perturbation $\theta=T-T_c$, with $T_c=\frac{1-y}{2}$, the forcing of the heat equation is $f_T$ itself.
Since $T_c$ is independent of $t$ and linear in $y$, $\partial_tT_c=\Delta T_c=0$, and the remaining term $\vu\cdot\nabla T_c=T_c'u_y$ is already part of \eqref{theta}.
The buoyancy of the conduction state, $G(T_c-T_{\rm ref})\ey$, is also kept in the momentum equation, so $f_y$ contains the full temperature $T$.
(The hydrostatic splitting of Section~\ref{sec:impl} is not used here.)

\paragraph{Initial data and steps} $\vu^0=0$, $p^0=0$, $T^0=T(\cdot,0)$ (all exact at $t=0$); the first step is backward Euler; $C=1$, $A=1$, $\eta=0.95$ (relaxation on), $Q^0=1$.

\paragraph{Time convergence (Table~\ref{tab:mms})} Uniform $40\times40$ mesh, $\dt=0.05\cdot2^{-k}$, $k=0,\dots,5$, and $L^2$ differences at $t=0.5$ from a reference run with $\dt=1.953125\cdot10^{-4}$ on the same mesh, the pressures compared after removing their means.
This isolates the time error, since the spatial error is the same in all runs.

\paragraph{Space convergence (Table~\ref{tab:space})} Uniform $N\times N$ meshes, $N=8,16,32,64$, $\dt=2.5\cdot10^{-4}$, errors at $t=0.5$ against the exact functions (not their interpolants), computed with a quadrature rule of order $9$.

\paragraph{Post-processing of Figure~\ref{fig:ra1e6}} The images are made from the computed velocity and temperature fields.
\begin{enumerate}
\item \emph{Sampling.} Every $750$ steps ($\Delta t=1.5\cdot10^{-3}$), the full temperature $T=T_c+\theta$ and the velocity are interpolated from the finite-element spaces to the nodes of a uniform $513\times129$ grid on $[0,4]\times[0,1]$.
The panels show the samples closest to $t=0.009$, $0.012$, and $0.03$.
\item \emph{Resampling.} Each field is interpolated by cubic splines, periodic in $x$, to an image of $2000\times500$ pixels.
\item \emph{Color.} The temperature is shown with the reversed `RdBu' color scheme (using Matplotlib), after the contrast stretch
\begin{equation*}
S(T)=\frac{1}{2}+\frac{1}{2}\,\frac{\tanh\big((T-\frac{1}{2})/s\big)}{\tanh\big(1/(2s)\big)},\qquad s=0.3 .
\end{equation*}
Here, $S$ is monotone and keeps $T=0$, $\frac{1}{2}$, and $1$ fixed.
This widens the range of colors near the bulk temperature $T=\frac{1}{2}$, where the well-mixed interior lies, and narrows it near the wall values.
\item \emph{Texture.} The streaky texture is a line integral convolution of the velocity \cite{CabralLeedom1993}.
White noise (with a fixed seed) is averaged along the streamlines of the unit vector field $\vu/|\vu|$, traced forward and backward from each pixel over $25$ steps of $0.6$ pixel with linearly decreasing weights, and normalized to zero mean and unit variance.
The color of each pixel is multiplied by $1+0.05\,\tau$, where $\tau$ is the texture.
The texture therefore shows the direction of the flow only, not its speed, and changes the brightness by a few percent.
\item \emph{Nusselt numbers.} The bottom panel plots the wall and volume Nusselt numbers of Section~\ref{sec:rolls}, recorded every $50$ steps by the solver.
(They are not post-processed.)
\end{enumerate}
The script that produces the figure from the run is part of the supplementary code.

\section*{Declaration of generative AI and AI-assisted technologies in the manuscript preparation process} The author discloses that Claude (Sonnet 4.6) was used during manuscript preparation to support language clarity, grammatical accuracy, and consistency in citation formatting.
Additionally, Claude (Opus 5.5) was used to organize, clean, and restructure parts of the supplemental code (for both FreeFem and \LaTeX), to develop the exact manufactured solution of \ref{app:mms}, and to help with post-processing scripts. After using this tool, the author reviewed and edited the content as needed. Outside of that, these tools were not used to generate the study's findings, analyses, or conclusions; all substantive intellectual contributions are those of the author, and the author takes full responsibility for the content of the published article.

\section*{Declaration of competing interest}
The author declares no known competing financial interests or personal relationships that could have appeared to influence the work reported in this paper.

\section*{Data availability}
The FreeFEM scripts, the command line of every run reported in this work, and the scripts that produce all tables and figures from the run logs are provided as supplementary material at: \url{https://github.com/mcmccurdy/SAV-for-RBC}.

\section*{Funding}
This research did not receive any specific grant from funding agencies in the public, commercial, or not-for-profit sectors.

\bibliographystyle{elsarticle-num}
\bibliography{refs}
\end{document}